\documentclass[conference,compsoc]{IEEEtran}

\ifCLASSOPTIONcompsoc
  \usepackage[nocompress]{cite}
\else
  \usepackage{cite}
\fi

\usepackage{cite}
\usepackage{amsmath,amsfonts}
\usepackage{amsthm}
\usepackage{algpseudocode}
\usepackage{graphicx}
\usepackage{textcomp}

\usepackage{url}
\usepackage{multirow}
\usepackage{hyperref}

\usepackage[ruled,vlined]{algorithm2e}
\SetKwInOut{Input}{Input}
\SetKwInOut{Output}{Output}
\SetKwBlock{Client}{Client:}{end}
\SetKwBlock{Server}{Server:}{end}

\usepackage{algpseudocode}
\usepackage{xcolor,soul}
\usepackage{makecell}
\usepackage{subcaption}

\usepackage{tikz}
\usetikzlibrary{shapes.geometric, arrows}

\usepackage{xcolor}

\definecolor{mypink}{HTML}{FB2E99}

\newcommand{\relu}{$\mathsf{ReLU}$}
\newcommand{\cmark}{\text{\ding{51}}}
\newcommand{\xmark}{\text{\ding{55}}}

\usepackage{tabularray}
\usepackage{pifont}

\newtheorem{theorem}{Theorem}
\newtheorem{lemma}{Lemma}

\usepackage{filecontents}

\begin{document}


\title{Efficient Arithmetic-and-Comparison Homomorphic Encryption with Space Switching}



\author{
Erwin Eko Wahyudi, Yan Solihin, Qian Lou$^{*}$\\
\textit{University of Central Florida}\\
\textit{\{wahyudierwin, yan.solihin, qian.lou\}@ucf.edu}
}

\maketitle

\makeatletter
\begingroup
\renewcommand{\thefootnote}{\fnsymbol{footnote}}
\renewcommand{\@makefntext}[1]{\noindent$^*$#1}
\footnotetext[1]{Qian Lou is the corresponding author: qian.lou@ucf.edu.}
\endgroup
\makeatother





\begin{abstract}

Fully homomorphic encryption (FHE) enables computation on encrypted data without decryption, making it central to privacy-preserving applications. However, no existing scheme efficiently supports both arithmetic and comparison operations in a unified framework. Prior approaches such as scheme switching and polynomial approximation face serious limitations: switching incurs prohibitive overhead for large inputs, while approximation methods introduce errors near critical points, restricting use in accuracy-sensitive tasks. We propose \emph{space switching} method to integrate arithmetic and comparison computation seamlessly within FV-style schemes. Our approach identifies that the two types of operations require different plaintext spaces and introduces two procedures: a reduction step to transition from the number space $\mathbb{Z}_{p^r}$ to the digit space $\mathbb{Z}_p$, and a modulus-raising step to map results back to $\mathbb{Z}_{p^r}$. This design enables continuous evaluation of arithmetic and comparison within the same scheme. Experiments show that our method achieves up to $17\times$ faster performance than scheme switching and $15\times$ faster than direct comparison on database workloads, demonstrating its practicality for real-world privacy-preserving computation. Code and artifacts are available at \url{https://github.com/UCF-Lou-Lab-PET/Universal-BGV}. 


\end{abstract}
\IEEEpeerreviewmaketitle
\section{Introduction}
\label{s:intro}

Fully homomorphic encryption (FHE) enables computation directly on encrypted data without decryption, making it a powerful primitive for secure computation. Since Gentry’s seminal construction~\cite{gentry2009fully}, numerous schemes have been proposed~\cite{brakerski2012fv,fan2012fv,brakerski2014bgv,ducas2015fhew,chillotti2020tfhe,cheon2017ckks}, each improving efficiency or expanding functionality, and FHE has been applied to privacy-preserving machine learning~\cite{gilad2016cryptonets,NEURIPS2019_56a3107c,lou2021hemet,lou2020glyph,lou2021safenet,lou2020falcon,feng2020cryptogru}, genomic data analysis~\cite{kim2015dna,zhang2015dna,raisaro2018dna}, and finance~\cite{armknecht2015finance,han2019finance}. Broadly, schemes fall into {\em word-wise} and {\em bit-wise} families: word-wise schemes such as BFV~\cite{brakerski2012fv,fan2012fv}, BGV~\cite{brakerski2014bgv}, and CKKS~\cite{cheon2017ckks} are highly efficient for arithmetic (e.g., matrix multiplication and convolution) and benefit from SIMD packing~\cite{smart2014simd}, but evaluating non-arithmetic functions like comparisons is very expensive. In contrast, bit-wise schemes such as FHEW~\cite{ducas2015fhew} and TFHE~\cite{chillotti2020tfhe} support logic operations via functional bootstrapping and Boolean/LUT circuits, yet cannot perform arithmetic operations efficiently (e.g., a 16-bit multiplication can take \mbox{$\sim$30 s}~\cite{lu2021pegasus}). This yields a dilemma: word-wise schemes excel at arithmetic operations but struggle with logic operations, while bit-wise excel at logic operations but struggle with arithmetic operations.

One of the most pressing logic operations is comparison, which is relied upon by many important applications. 
For example, a single database query may involve a mixture of comparison and arithmetic operations. Figure~\ref{fig:db workflow} illustrates an example query pipeline where ciphertexts are compared (via the WHERE clause) and then multiplied and summed to produce encrypted query results.~\footnote{While some specialized solutions exist~\cite{tan2020efficientsefv,boneh2013PrivateDatabaseQueries,kim2016betterSecurityForQueries,kim2016efficiencysefv} which show how word-wise FHE can filter (equality tests) and return exact matches, they are largely limited to retrieval-style queries and do not directly support aggregations such as \texttt{SUM}, which are essential for analytics.} Similarly, in privacy-preserving machine learning applications, a single inference often involves both arithmetic operations (convolutions and matrix multiplications) and comparisons (in activation layers like \relu).

\begin{figure*}[htbp]
    \centering
    \includegraphics[width=\linewidth]{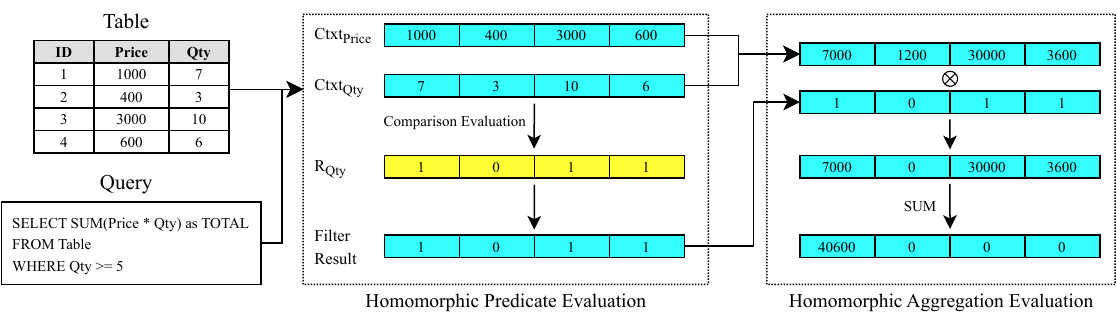}
    \caption{Database queries inherently combine comparison (e.g., filtering with predicates) and arithmetic (e.g., multiplication and aggregation), requiring both to be supported efficiently under homomorphic encryption.}
    \label{fig:db workflow}
\end{figure*}

Some specialized workarounds exist, for example a comparison can be approximated using polynomials which allows converting it into arithmetic operations. However, an approximation is not always effective. For example, in genomics applications~\cite{kim2015dna,zhang2015dna,raisaro2018dna}, exact comparisons are critical, since small errors can lead to incorrect conclusions. They may also result in non-negligible accuracy loss especially if low degree polynomials are used~\cite{gilad2016cryptonets}. Even with high-degree polynomials, approximation methods (with CKKS) are generally valid only over small intervals, such as $[-1,1]$~\cite{cheon2020approx,lee2021minimaxapprox}, and are inherently inaccurate near critical points such as zero. 

A more general-purpose approach to enabling mixed arithmetic and comparison uses homomorphic \emph{scheme switching}~\cite{boura2020chimera,lu2021pegasus}, where arithmetic is performed with word-wise schemes and comparison with bit-wise schemes, or \emph{functional bootstrapping}~\cite{lee2024functional,alexandru2025generalfunctionalbootstrapping}, where comparisons are performed via lookup tables (LUTs). While these approaches leverage the strengths of each category, scheme conversions and LUT evaluations remain prohibitively costly. Because the complexity grows exponentially with input bit-width, they typically restrict the bit-width~\cite{alexandru2025generalfunctionalbootstrapping}, support only narrow intervals such as $[-8,8]$~\cite{cheon2022domainextension}, and may require additional bootstrapping~\cite{al2022openfhe}.

Thus, we propose a new approach to address these limitations. The main idea is that comparing two numbers can be decomposed into multiple comparisons of their digits and their aggregation. Our key observation is that performing comparison at the digit level is much more efficient homomorphically because we can reduce the ciphertext {\em number space} to a much smaller {\em digit space}. Thus, in our {\em space switching} approach, we convert numbers from the number space $\mathbb{Z}_{p^r}$ to their digits, each in the base-$p$ digit space  $\mathbb{Z}_p$ (where $p^r$ represents the plaintext modulus). This greatly reduces complexity, since we deal with much smaller polynomials (of degree $p-1$ rather than $p^r-1$). Effectively, we transform a number-space comparison that incurs $O\left(p^{r/2}\right)$ multiplications, into digit-space total of $O\left(r \sqrt{pr}\right)$ multiplications. This advantage grows as the numbers being compared are enlarged. For example, if $r$ is fixed at 5, then as $p$ grows, the gap in the number of multiplications is quadratic to $p$, i.e. $O\left(p^2\sqrt{p}\right)$ vs. $O\left(\sqrt{p}\right)$


Practically, our space switching method offers multiple advantages over scheme switching. It avoids costly cross-scheme conversions between word-wise and bit-wise HE, and supports large bit-widths and arbitrary ranges, while at the same time retains the SIMD benefit of the BGV/BFV schemes. It does not require special encodings and works fully homomorphically end to end. 



We integrated our \emph{space switching} into a production-grade library (HElib), enabling {\em end-to-end} performance evaluation that includes all the overheads including those from switching between number and digit spaces. We evaluate across multiple bit-widths and representative workloads against two baselines that support arithmetic and exact comparison. For primitive operators (e.g., \texttt{LT}, \texttt{EQ}), our method achieves up to $130\times$ speedup over Morimura \emph{et~al.}~\cite{morimura2023accelerating} and up to $40\times$ speedup over scheme switching~\cite{bian2023he3db} in the amortized (per-slot) sense. In a database workload, \emph{space switching} achieves up to $17\times$ speedup over scheme switching~\cite{bian2023he3db} and up to $15\times$ over Morimura \emph{et~al.}~\cite{morimura2023accelerating}. We report complete configurations and breakdowns in Sections~\ref{s:expt}--\ref{s:results}.

To summarize, this paper makes the following contributions:
\begin{itemize}
  \item We present a novel \emph{space switching} approach that enables mixed arithmetic and comparison operations by transitioning between number and digit spaces. 
  
  \item We propose an efficient digit extraction technique to convert from number to digit space and a modulus raising technique to convert from digit to number space. The digit extraction technique is adapted from BGV/BFV bootstrapping but we adapt it in order order to remove its inefficiency. 
  
  \item We integrate our technique into the production-grade FHE library HElib, and evaluate primitive operators and a database workload across multiple bit-widths. We show that our space switching technique achieves huge speedups over the state-of-the-art scheme switching. 
\end{itemize}

\section{Related Works}

Table~\ref{tab:related work} organizes prior efforts on enabling logic (we focus on comparison) operations in homomorphic encryption into distinct categories, clarifying their underlying approaches and trade-offs. The table highlights key properties, including support for SIMD packing, arithmetic capability, exactness of comparison evaluation, scalability to large inputs, and the underlying scheme. SIMD packing is important to amortize a long latency computation over many operations, and is a key advantage of word-wise FHE schemes over bit-wise ones. Arithmetic capability is important for most FHE applications but some FHE schemes do not support them. Exactness of computation is important for many applications, especially when comparisons are involved. Scalability to large input determines how fast computation can be performed on larger bit widths; small bit widths restrict the applicability to a narrower set of applications. The table allows us to compare very different lines of work on a common basis and to identify where existing techniques fall short.

\begin{table*}[htbp]
\centering
\small
\caption{Comparison of prior works supporting comparison operations in FHE.}
\begin{tblr}{
    colspec = {| c | c | c c c c | c |},
    row{1} = {font=\fontsize{8}{11}\bfseries,rowsep=1pt},
    row{2-Z} = {font=\fontsize{8}{11},rowsep=1pt},
    }
\hline
\textbf{Approach} & \textbf{Name} & \textbf{SIMD} & \textbf{Arithmetic} & \textbf{Exact} & {\textbf{Bit-width scalability}} & \textbf{Scheme}\\
\hline\hline
\SetCell[r=2]{c}Bit-wise & TFHE~\cite{chillotti2020tfhe} & \xmark & \xmark & \cmark & {binary-circuit depth} & TFHE  \\
 & Chakraborty~\cite{chakraborty2022tfhecmp} & \xmark & \xmark & \cmark & {binary-circuit depth} & TFHE \\\hline
\SetCell[r=3]{c}Scheme Switching & Chimera~\cite{boura2020chimera} & \xmark & \cmark & \cmark & {binary-circuit depth} & BFV/CKKS/TFHE\\
Scheme Switching & Pegasus~\cite{lu2021pegasus} & \xmark & \cmark & \cmark & {binary-circuit depth} & BFV/CKKS/TFHE\\
Scheme Switching & HE$^3$DB~\cite{bian2023he3db} & \xmark & \cmark & \cmark & {binary-circuit depth} & BFV/CKKS/TFHE\\\hline
\SetCell[r=2]{c}Approximation & Cheon \textit{et al.}~\cite{cheon2020approx} & \cmark & \cmark & \xmark & {approximation polynomial degree} & CKKS\\
Approximation & Lee \textit{et al.}~\cite{lee2021minimaxapprox} & \cmark & \cmark & \xmark & {approximation polynomial degree} & CKKS\\\hline
\SetCell[r=3]{c}Interpolation & Tan \textit{et al.}~\cite{tan2020efficientsefv} & \cmark & \xmark & \cmark & {$O(r\sqrt{p})$} & BFV/BGV\\
Interpolation & Iliashenko \textit{et al.}~\cite{iliashenko2021fasterbgvsefv} & \cmark & \xmark & \cmark & {$O(r\sqrt{p})$} & BFV/BGV\\
Interpolation & Morimura \textit{et al.}~\cite{morimura2023accelerating} & \cmark & \cmark & \cmark & {$O(p^{r/2})$} & BFV/BGV\\\hline
\SetCell[r=3]{c}Functional Bootstrapping & Liu \textit{et al.}~\cite{liu2023amortizedfunctionalbootstrapping} & \cmark & \cmark & \cmark & {LUT/domain size} & BFV/BGV\\
 & Lee \textit{et al.}~\cite{lee2024functional} & \cmark & \cmark & \cmark & {LUT/domain size} & BFV/BGV/CKKS\\
 & Alexandru \textit{et al.}~\cite{alexandru2025generalfunctionalbootstrapping} & \cmark & \cmark & \xmark & {LUT/domain size} & CKKS\\\hline
\SetCell[r=1]{c}Others & XCMP~\cite{lu2018non} & \xmark & \xmark & \cmark & {input-domain size} & BFV/BGV\\
\hline
\SetCell[c=2]{c}\textbf{Ours (Space Switching)} & & \cmark & \cmark & \cmark & {$O(r\sqrt{pr})$} & BFV/BGV\\
\hline
\end{tblr}
\label{tab:related work}
\end{table*}

{\bf Bit-wise approach}.
Bit-wise schemes~\cite{ducas2015fhew,chillotti2020tfhe,chakraborty2022tfhecmp} 
in general have very low arithmetic performance. For instance, multiplying two 16-bit integers in TFHE can take up to 30 seconds~\cite{lu2021pegasus}. While LUT-based bootstrapping enables accurate comparison, relying solely on bit-wise FHE becomes impractical for workloads dominated by arithmetic~\cite{lou2019she, lou2021safenet, jiang2022matcha, lou2021hemet}.

{\bf Scheme switching}.
Homomorphic scheme switching performs arithmetic operations with word-wise schemes and switches to bit-wise schemes to perform comparison~\cite{boura2020chimera,lu2021pegasus}. The switching time is very costly and increases exponentially with the input bit-width, hence they typically limit the bit-width and only support narrow intervals such as $[-8,8]$~\cite{cheon2022domainextension}. Each switch is often followed by bootstrapping~\cite{al2022openfhe}, which is one of the most expensive operations in word-wise HE and is generally avoided in practice~\cite{iliashenko2021fasterbgvsefv, lou2019glyph, zhang2023hebridge, xue2026sok}.

{\bf Polynomial approximation}. 
For the CKKS scheme, {\em polynomial approximation} is a popular method, which approximates a comparison function with polynomials. Several works~\cite{cheon2019approx,cheon2020approx,lee2021minimaxapprox,zhang2025cipherprune, zhang2024heprune} approximate comparison functions using either high-degree polynomials or compositions of lower-degree ones. The comparison function is fundamental, as it forms the basis for many comparison primitives such as step and sign functions, and can also be used to construct LUTs for arbitrary non-arithmetic functions~\cite{cheon2024tree}. However, these approximation methods are generally valid only over small intervals, such as $[-1,1]$~\cite{cheon2020approx,lee2021minimaxapprox}, and are inherently inaccurate near critical points such as zero. Achieving higher accuracy and precision generally requires higher polynomial degree and larger CKKS scales, increasing depth, noise, and runtime, yet a zero-margin cannot be certified, so exact correctness is not guaranteed. Such errors make them unsuitable for applications requiring exact correctness, including genomics~\cite{kim2015dna,zhang2015dna,raisaro2018dna} and finance~\cite{han2019finance,armknecht2015finance}.

{\bf Polynomial interpolation}. 
For BFV/BGV schemes, exact comparison can be supported via {\em polynomial interpolation}. Prior works~\cite{tan2020efficientsefv,iliashenko2021fasterbgvsefv} have demonstrated that comparison functions can be expressed as interpolation polynomials over integers modulo $p$. Moreover, since FV-style schemes support SIMD, the cost of these comparisons can be amortized across many slots, promising to yield per-integer latencies comparable to or even better than those in bit-wise schemes~\cite{iliashenko2021fasterbgvsefv}. However, the state-of-the-art comparison method~\cite{iliashenko2021fasterbgvsefv, yudha2024boostcom} relies on special plaintext encodings that prevent interpolation-based comparison from being used together with arithmetic operations in the same scheme~\cite{cheon2020approx,morimura2023accelerating}. Another potential approach is to construct comparison polynomials directly in the native plaintext number space $\mathbb{Z}_{p^r}$. However, the degree of such polynomials grows exponentially with $r$ (up to $p^r$), making direct comparison evaluation in $\mathbb{Z}_{p^r}$ infeasible.

\textbf{Functional bootstrapping}.
Recently, several functional bootstrapping methods~\cite{alexandru2025generalfunctionalbootstrapping,lee2024functional,liu2023amortizedfunctionalbootstrapping} have been proposed to allow arbitrary function computation during the bootstrapping process. Functional bootstrapping refreshes ciphertexts and applies a chosen function simultaneously, effectively removing noise while completing an LUT evaluation. Specifically, Lee \textit{et al.}~\cite{lee2024functional} proposed a functional bootstrapping technique that takes an RLWE ciphertext, i.e., CKKS or BFV, as input and outputs the refreshed BFV ciphertext, building upon BFV-style bootstrapping. However, such BFV-style bootstrapping is far from efficient; for instance, it can take about 3 minutes for a single evaluation due to the intrinsically inefficient slot utilization of BFV-style bootstrapping~\cite{alexandru2025generalfunctionalbootstrapping}. To address this issue, Alexandru \textit{et al.}~\cite{alexandru2025generalfunctionalbootstrapping} presented a general functional bootstrapping technique based on CKKS-style bootstrapping, which has the best throughput among all FHE methods, to improve amortized performance. However, this method suffers from high computational complexity for large input spaces. Their experimental results are limited to 12-bit LUTs because the scaling factor nears the limit for 64-bit modular operations, and the complexity of polynomial evaluation increases significantly, i.e., approximately 1 minute per evaluation for 9-bit LUTs, but about 10 minutes for 12-bit LUTs.

\textbf{Others}.
There are specialized methods targeting comparison, but they are limited in generality or efficiency. For example, XCMP~\cite{lu2018non} encodes numbers into exponents to facilitate comparison, but this design precludes SIMD packing and does not support arithmetic.

{\bf Overall}, our space switching enables arithmetic and comparison operations within a single FHE scheme (BFV/BGV), hence it retains its SIMD advantage. Due to a much higher efficiency, it supports exact comparison even with large inputs.

\section{Background}
\label{s:background}


\subsection{The BGV Scheme}

The BGV scheme is a lattice-based cryptographic construction built on the Ring Learning with Errors (RLWE) assumption~\cite{brakerski2014bgv}. Although mathematically complex, RLWE is a fundamental problem that underpins the security of the scheme. The main parameters of BGV are summarized in Table~\ref{tab:bgv_parameters}, including $p$, $r$, $m$, $n$, and $Q$. The value $p^r$ specifies the plaintext modulus, whereas $Q$ defines the ciphertext modulus. In practice, $Q$ is chosen much larger than $p$


\begin{table}[htbp]
    \centering
    \caption{Parameters used in BGV.}
    \label{tab:bgv_parameters}
    \begin{tblr}{
    colspec = {| c | l |},
    row{1} = {font=\fontsize{8}{11}\bfseries,rowsep=1pt},
    row{2-Z} = {font=\fontsize{8}{11},rowsep=1pt},
    }
    \hline
        \textbf{Parameter} & \textbf{Description} \\ \hline\hline
        $p^r$ & Plaintext coefficient modulus. \\
        $m$ & Order of the cyclotomic ring. \\
        $n$ & Degree of the cyclotomic polynomial. \\
        $L$ & Maximum multiplicative level. \\
        $Q$ & Product of prime moduli: $Q = \prod_{i=0}^L q_i$. \\
        $\lambda$ & Security parameter of a BGV instance. \\
        \hline
    \end{tblr}
\end{table}

In this setting, \( \Phi_m(x) \) denotes the $m$-th cyclotomic polynomial of degree $n$. The relation between $m$ and $n$ is determined by the Euler totient function \( \phi \), with \( n = \phi(m) \). Earlier works often restricted $n$ to powers of two for simplicity, but more recent research suggests that non-power-of-two values of $n$ yield better performance and stronger security flexibility~\cite{Gentry2012FHEPolylog, iliashenko2021fasterbgvsefv}. The polynomial ring \( R_{p^r} \) is defined as \( \mathbb{Z}_{p^r}[x] / (\Phi_m(x)) \), consisting of polynomials modulo $\Phi_m(x)$ with coefficients in \( \mathbb{Z}_{p^r} \). Likewise, \( R_Q \) is defined with modulus $Q$. In the BGV scheme, plaintexts are encoded in \( R_{p^r} \), while ciphertexts are represented in \( R_{Q} \).

{Let $d=\textsf{Ord}(m,p)$ denote the multiplicative order of $p$ modulo $m$, i.e. the smallest positive integer $d$ such that $p^d \equiv 1 \pmod{m}$.} According to the algebraic structure described in~\cite{smart2014simd}, the cyclotomic polynomial \( \Phi_m(x) \) can be decomposed into $\ell$ irreducible factors as $\Phi_m(x) = \prod_{i=1}^\ell F_i(x) \pmod{p^r}$,
where $\ell = \phi(m)/d$ and each $F_i(x)$ has degree $d$. Consequently, the quotient ring can be represented as
\[R_{p^r} = \mathbb{Z}_{p^r}[x]/\Phi_m(x) \;\cong\; \mathbb{Z}_{p^r}[x]/F_1(x) \otimes \cdots \otimes \mathbb{Z}_{p^r}[x]/F_\ell(x).\]
{That means a plaintext slot is an element of the extension field
$\mathbb{Z}_{p^r}[x]/F_i(x)$, represented as a polynomial of degree at
most $d$, which may in particular be a constant value in $\mathbb{Z}_{p^r}$. The parameter $\ell$ also corresponds to the number of
available plaintext slots in BGV. These slots enable SIMD-style
packing, where each slot represents one plaintext value.
Homomorphic addition and multiplication are performed on the
ciphertext in $R_Q$, while the corresponding slot values are
recovered in the plaintext space after decryption.}


In the BGV scheme, the message $\mathbf{m}$ is encoded as a polynomial from $R_{p^r}$. The secret key \( \mathbf{s} \in R_Q \) is a {``small''} polynomial, with each coefficient drawn from the set $\{-1, 0, 1\}$. The ciphertext corresponding to $\mathbf{m}$ is $c=(\mathbf{c}_0,\mathbf{c}_1)$ with $\mathbf{c}_0,\mathbf{c}_1\in R_Q$, satisfying
\begin{equation}
\label{eq: bgv}
\mathbf{c}_0 + \mathbf{c}_1 \cdot \mathbf{s} \;\equiv\; \mathbf{m} + p^r \mathbf{e} \pmod{Q}.
\end{equation}

Having the plaintext $\mathbf{m}$ and a secret key $\mathbf{s}$, the encryption process begins by sampling a polynomial \( \mathbf{c}_1 \) from \( R_Q \) uniformly. The polynomial \( \mathbf{c}_0 \) in \( R_Q \) is then computed based on \( \mathbf{c}_1 \) and an additional random polynomial $\mathbf{e}$, which is generated from a predefined Gaussian distribution. This computation results in the ciphertext pair \( \mathbf{c}_0 \) and \( \mathbf{c}_1 \), as shown in Algorithm~\ref{alg:encryption_algo}. Decryption is performed by Algorithm~\ref{alg:decryption_algo}: given a ciphertext $c=(\mathbf{c}_0, \mathbf{c}_1)$, it recovers the plaintext $\mathbf{m}$.

{The parameters $L$ and $Q$ determine the supported multiplicative depth in FHE, i.e., the maximum number of sequential multiplications that can be applied to a ciphertext along any computation path. Multiplicative depth counts the longest chain of sequential multiplications, rather than the total number of multiplications. For example, computing $R = (ab)(cd)$ requires three multiplications in total, but its multiplicative depth is only two: one level to compute $R_1 = ab$ and $R_2 = cd$, and a second level to compute $R = R_1R_2$. If the multiplicative depth of a ciphertext exceeds the supported level $L$, then decryption becomes invalid, meaning that the decrypted result no longer matches the original plaintext.}

\begin{algorithm}[t]
\caption{BGV Encryption}
\label{alg:encryption_algo}
\SetAlgoLined
\LinesNumbered
\KwIn{plaintext $\mathbf{m} \in R_{p^r}$, secret key $\mathbf{s} \in R_Q$}
\KwOut{ciphertext $c=(\mathbf{c}_0,\mathbf{c}_1) \in R_Q^2$}
$\mathbf{c}_1 \leftarrow \text{sample from }(R_Q)$\;
$\mathbf{e} \leftarrow \text{sample from } R_p$\;
$\mathbf{c}_0 \leftarrow \mathbf{m} + p^r \cdot \mathbf{e} - \mathbf{c}_1 \cdot \mathbf{s}$\;
\Return $c=(\mathbf{c}_0,\mathbf{c}_1)$\;
\end{algorithm}

\begin{algorithm}[t]
\caption{BGV Decryption}
\label{alg:decryption_algo}
\SetAlgoLined
\LinesNumbered
\KwIn{ciphertext $c=(\mathbf{c}_0,\mathbf{c}_1) \in R_Q^2$, secret key $\mathbf{s} \in R_Q$}
\KwOut{plaintext $\mathbf{m} \in R_{p^r}$}
$\mathbf{m}^* \leftarrow \mathbf{c}_0 + \mathbf{c}_1 \cdot \mathbf{s}$\;
$\mathbf{m} \leftarrow \mathbf{m}^* \bmod p^r$\;
\Return $\mathbf{m}$\;
\end{algorithm}

\subsection{Polynomial Interpolation}
The non-arithmetic operations are typically evaluated via polynomial interpolation in the field $\mathbb{Z}_p$ in the FV-style HE schemes. The interpolation polynomials are key to the comparison function~\cite{tan2020efficientsefv, iliashenko2021fasterbgvsefv} as well as other non-arithmetic functions like modulo, Hamming weight and division~\cite{iliashenko2021mod, morimura2023accelerating}. The construction of these interpolation polynomials relies on important theorems and lemmas that we will describe below.
\begin{theorem}[Fermat’s Little Theorem on $\mathbb{Z}_p$]
\label{lm: fermat}
Let $p$ be a prime, for any $a\in \mathbb{Z}_p \setminus \{0\}$, we have 
\[
a^{p-1} = 1 \pmod{p}.
\]
\end{theorem}

With Fermat’s Little Theorem in Theorem \ref{lm: fermat}, we can evaluate the equality function easily over $\mathbb{Z}_p$. The equality check for $x,y \in \mathbb{Z}_p$ can be computed as

\begin{equation}
\label{eq: equality function}
F_{EQ}(x,y) = 1-(x-y)^{p-1} = 
\begin{cases} 
    1 & \text{if } x = y \\
    0 & \text{otherwise}
\end{cases}.
\end{equation}

Based on the equality function, any multi-variable function can be evaluated via a multi-variate interpolation polynomial according to the following lemma:
\begin{lemma}
\label{lm: interpolate}
Every $n$-variable function $f:\mathbb{Z}_p^n \rightarrow \mathbb{Z}_p$ is a polynomial function represented by a $n$-variate polynomial $P_f(X_1,...,X_n)$ of degree at most $p-1$ in each variable:
\[
P_f(X_1,...,X_n)=\sum_{\mathbf{a}\in\mathbb{F}_p^n}f(\mathbf{a})\prod_{i=1}^n(1-(X_i-a_i)^{p-1}),
\]
where $a_i$ is the $i$-th coordinate of the length-$n$ vector $\mathbf{a}$.
\end{lemma}

With Lemma \ref{lm: interpolate}, we know that for a uni-variate function $f:\mathbb{Z}_p \rightarrow \mathbb{Z}_p$, the uni-variate interpolation polynomial has the following form:
\[
P_f(X) = \sum_{a=0}^{p-1} f(a)\bigl(1 - (X-a)^{p-1}\bigr).
\]

The above interpolation polynomial has at most degree $(p-1)$. Such a univariate polynomial of degree $(p-1)$ can be computed with at most $\sqrt{p-1}$ non-scalar multiplications thanks to the following theorem and the Paterson-Stockmeyer algorithm:
\begin{theorem}[from ~\cite{paterson1973number}]
\label{lm: paterson}
Any polynomial of degree $d$ over a ring can be evaluated using $O(\sqrt{d})$ non-scalar multiplications and $\lceil \log_2 d \rceil + 1$ multiplicative levels.
\end{theorem}

We illustrate the interpolation process using the less-than-zero ($LT0$) function over $\mathbb{Z}_p$, defined as
\begin{equation}
\label{eq: lt values}
LT0(x) =  
\begin{cases} 
    1 & \text{if } x < 0, \\
    0 & \text{otherwise}.
\end{cases}
\end{equation}
The mapping between inputs and outputs for the $LT0$ function is shown in Table~\ref{tab:lt0 values}.

\begin{table}[h]
    \centering
    \caption{The truth table for the $LT0$ function over $\mathbb{Z}_p$.}
    \label{tab:lt0 values}
        \begin{tblr}{
    colspec = {|c||c|c|c|c|c|c|c|c|},
    row{1} = {font=\fontsize{8}{11},rowsep=1pt},
    row{2-Z} = {font=\fontsize{8}{11},rowsep=1pt},
    }
        \hline
       $x$  &  $-\frac{p-1}{2}$ & $\ldots$ & -1 & \textbf{0} & 1 & $\ldots$ & $\frac{p-1}{2}$ \\\hline
       $LT0(x)$  & 1 & $\ldots$ & 1 & \textbf{0} & 0 & $\ldots$ & 0 \\\hline
    \end{tblr}
\end{table}

According to Theorem~\ref{lm: fermat}, we can interpolate $LT0$ by summing the equality polynomials for all negative inputs:
\[F_{LT0}(x) = \sum_{i=-(p-1)/2}^{-1} F_{EQ}(x, i).\]
By Lemma~\ref{lm: interpolate}, the resulting polynomial $F_{LT0}(x)$ has degree $p-1$, which can be evaluated in $\sqrt{p-1}$ non-scalar multiplications. This construction shows how the $LT0$ function can be expressed as a polynomial suitable for homomorphic evaluation.

\section{FV-Style Comparison Challenges}
\label{s:challenges}

In this section, we present a systematic analysis of the challenges of performing comparisons, such as less-than (LT) and equality (EQ), in FV-style schemes. Section~\ref{sec:method limit logic fv} establishes inherent limitations for comparisons over $\mathbb{Z}_{p^r}$, and Section~\ref{sec:method limit special ptxt encoding} reviews special plaintext encodings for comparisons proposed in~\cite{tan2020efficientsefv,iliashenko2021fasterbgvsefv}, highlighting their incompatibility with standard arithmetic within the same scheme.


\subsection{Limitations of Comparison in $\mathbb{Z}_{p^r}$}
\label{sec:method limit logic fv}

The BGV and BFV schemes support SIMD-style computation by utilizing slots, where each slot contains an integer $a \in \mathbb{Z}_{p^r}$. Under this representation, arithmetic operations such as addition and multiplication are well-defined and yield the correct results modulo $p^r$.

The situation is different for comparison. For example, consider the equality test. To determine whether two integers $a, b \in \mathbb{Z}_{p^r}$ are equal, one can equivalently test whether $a-b=0$. This motivates the definition of the equality function $EQ$, where $EQ(x,0)=1$ for $x=0$ and $EQ(x,0)=0$ otherwise. Over $\mathbb{Z}_p$ (i.e., $r=1$), the $EQ$ function can be expressed by a polynomial of degree $p-1$ as shown in Eq.~\eqref{eq: equality function}. However, over $\mathbb{Z}_{p^r}$ with $r>1$, two challenges arise. First, a direct comparison of $EQ$ over the entire domain $\mathbb{Z}_{p^r}$ requires a polynomial of degree $p^r-1$, which is infeasible in practice due to the exponential growth in degree and the corresponding multiplicative depth. Second, by the theory of polyfunctions~\cite{geelen2023boot}, for $r>1$ the equality predicate over $\mathbb{Z}_{p^r}$ is not a polyfunction, meaning that it cannot be represented as a polynomial with integer coefficients over $\mathbb{Z}_{p^r}$.

This limitation is not restricted to equality. The less-than (LT) function over $\mathbb{Z}_{p^r}$, defined by
\begin{equation*}
\mathrm{LT}(x,0) =
\begin{cases}
    1 & \text{if } x < 0,\\
    0 & \text{otherwise},
\end{cases}
\end{equation*}
suffers the same challenges as the EQ function. Therefore, both EQ and LT in the domain $\mathbb{Z}_{p^r}$, particularly for $r>1$, are nontrivial and cannot be performed using polynomial interpolation. This reveals a fundamental obstacle to supporting comparison directly in the $\mathbb{Z}_{p^r}$ domain and motivates the use of the $\mathbb{Z}_p$ domain, where such comparisons can be expressed via interpolation polynomials over $\mathbb{Z}_p$.

\subsection{Limitations of Plaintext Encoding for Comparison}
\label{sec:method limit special ptxt encoding}

Prior works have proposed special encoding techniques to enable comparison on large numbers~\cite{tan2020efficientsefv, iliashenko2021fasterbgvsefv}. For a number $a \in \mathbb{Z}_{p^r}$, the first step is to decompose $a$ into its base-$p$ digits $(a_0, a_1, \ldots, a_{r-1}) = \sum_{i=0}^{r-1} a_i p^i$, where each digit $a_i$ lies in $\mathbb{Z}_p$. This decomposition is carried out in cleartext before encryption. The digits are then encoded as a polynomial over the extension field $\mathbb{F}_{p^d}$, namely $\sum_{i=0}^{d-1} a_i X^i$, which serves as the encoded plaintext. As a result, each ciphertext slot holds a polynomial that compactly represents an entire multi-digit number. To perform comparisons, the coefficients of $\sum_{i=0}^{d-1} a_i X^i$ are extracted homomorphically~\cite{iliashenko2021fasterbgvsefv}, producing ciphertexts corresponding to the individual digits $a_0, a_1, \ldots, a_{r-1}$ with $a_i \in \mathbb{Z}_p$. Once the digits are extracted in encrypted form, comparison polynomials defined over $\mathbb{Z}_p$ can be applied to evaluate functions such as equality or less-than.

Although this special plaintext encoding supports comparison, it does not naturally extend to general arithmetic operations such as multiplication and addition that produce integers larger than the modulus. This limitation arises because numbers are expressed in vector form, and the arithmetic between specially encoded numbers follows the rules of polynomial rings rather than integers. For example, consider $p=5$ and $d=3$. Integers can be decomposed into base-$5$ digits. Suppose $x, y \in \mathbb{Z}_{5^3}$ with $x=33$ and $y=19$. Under the special encoding, $x$ is represented by its base-$5$ digits $x_i \in \mathbb{Z}_5$, i.e., $\{x_0, x_1, x_2\} = \{3,1,1\}$, and encoded as $3X^0 + 1X^1 + X^2$. Similarly, $y$ is encoded as $4X^0 + 3X^1 + 0X^2$. Arithmetic on these encodings is carried out in the polynomial ring $\mathbb{Z}_5[X]/(X^3+1)$. Adding $x$ and $y$ yields $x+y = 2X^0 + 4X^1 + X^2$. Decoding this result gives $47$, which does not match the expected integer addition result of $33+19=52$ modulo $5^3$.

In summary, the special encoding is effective for applying comparisons over digits, but it is not compatible with standard integer arithmetic within the same ciphertext layout. Mixing arithmetic with such encodings would require re-encoding or additional conversions, which breaks the simplicity of the word-wise workflow. This motivates a different approach that keeps plaintexts as integers in $\mathbb{Z}_{p^r}$ for arithmetic, temporarily reduces to $\mathbb{Z}_p$ only to evaluate comparisons, and then lifts the result back to $\mathbb{Z}_{p^r}$, which we detail in Section~\ref{s:method}.

\section{Method}
\label{s:method}

This section presents our proposed method for interleaving arithmetic and comparison within a single FV-style scheme. {The main challenge is connecting the number space $\mathbb{Z}_{p^r}$ used for arithmetic and the digit space $\mathbb{Z}_{p}$ used for comparison.} Section~\ref{sec:method mod} specifies the BGV extension that enables arithmetic over $\mathbb{Z}_{p^r}$ along with comparison evaluation. Section~\ref{sec:method decomp} details a homomorphic reduction from the original plaintext space $\mathbb{Z}_{p^r}$ to the digit space $\mathbb{Z}_p$. Comparisons are then performed efficiently in the digit space $\mathbb{Z}_p$ via polynomial interpolation (Section~\ref{sec:method befv}). Finally, Section~\ref{sec:method lift} describes a modulus-raising step that restores the plaintext space from the digit space $\mathbb{Z}_{p}$ back to $\mathbb{Z}_{p^r}$, allowing continuous computation.

\subsection{Connecting Arithmetic and Comparison}
\label{sec:method mod}

Section~\ref{sec:method limit logic fv} shows that comparisons in the plaintext space $\mathbb{Z}_{p^r}$ are nontrivial when $r>1$. A natural alternative is to perform both arithmetic and comparison in $\mathbb{Z}_p$~\cite{morimura2023accelerating}. However, efficient evaluation of interpolation polynomials over $\mathbb{Z}_p$ typically requires small primes, which limits the practicality for large inputs. For example, $p$ ranges from $2$-$7$ in~\cite{tan2020efficientsefv}, is at most $173$ in~\cite{iliashenko2021fasterbgvsefv}, and at most $257$ in~\cite{morimura2023accelerating}. Restricting all computation to $\mathbb{Z}_p$ is therefore either unable to handle large numbers or incurs exponential overhead as the input bit-width increases. This motivates an alternative approach that separates arithmetic and comparison into their natural domains, $\mathbb{Z}_{p^r}$ for arithmetic and $\mathbb{Z}_p$ for comparison, while providing an efficient conversion mechanism between them.

\begin{figure}[htbp]
    \centering
    \includegraphics[width=\linewidth]{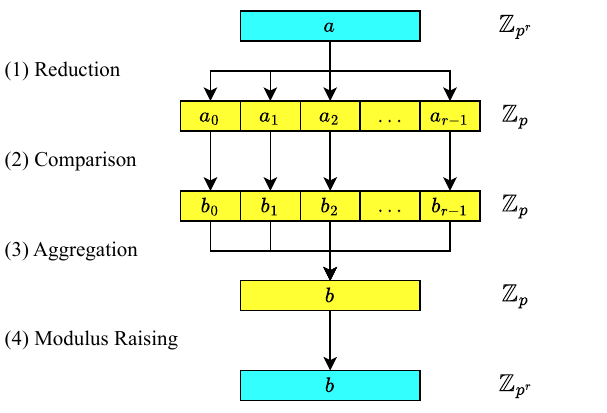}
    \caption{Overview of continuous evaluation of arithmetic operations over $\mathbb{Z}_{p^r}$ and comparisons over $\mathbb{Z}_p$.}
    \label{fig:workflow}
\end{figure}

As illustrated in Figure~\ref{fig:workflow}, the proposed method enables both arithmetic and comparison operations within the BGV scheme. We start with an integer $a \in \mathbb{Z}_{p^r}$, which may be either a raw input or the result of a prior arithmetic computation. To prepare $a$ for comparison, we first apply the reduction step, which produces the digit decomposition $\{a_i\}_{i=0}^{r-1}$ with each $a_i \in \mathbb{Z}_p$. Once represented in the digit space $\mathbb{Z}_p$, comparison functions such as equality and less-than can be efficiently evaluated on a per-digit basis using polynomial interpolation, yielding intermediate results $b_i$. These per-digit outputs are then aggregated into a final comparison result $b$. Since $b$ resides in the digit space $\mathbb{Z}_p$, we apply modulus raising so that $b$ is expressed in $\mathbb{Z}_{p^r}$, ensuring the result is usable in subsequent arithmetic computations.

Figure~\ref{fig:workflow} is straightforward to implement in plaintext. We now describe how the workflow in Figure~\ref{fig:workflow} is realized in ciphertext. {We start from a single ciphertext in $R_Q$ with $\ell$ packed plaintext slots, where each slot represents a plaintext value in $\mathbb{Z}_{p^r}$.} The reduction stage produces $r$ new ciphertexts, where every slot in the $i$-th ciphertext represents the $i$-th base-$p$ digit of the original slot value. Slot positions are preserved, so slot $k$ in every reduced ciphertext corresponds to the same original slot $k$. Comparison is then evaluated digit by digit within the digit space $\mathbb{Z}_p$. For each digit position $i$, we evaluate the predicate on the $i$-th reduced ciphertext, producing an indicator ciphertext whose slots contain the per-slot result for digit $i$. These $r$ indicator ciphertexts are then aggregated homomorphically into a single packed ciphertext whose slots contain the final comparison result (for example, $0$ or $1$) for each original slot. Finally, modulus raising re-embeds that packed comparison result from $\mathbb{Z}_p$ into $\mathbb{Z}_{p^r}$ without changing the packing layout, so it can be used immediately by subsequent arithmetic. All steps act slotwise and preserve SIMD packing.

This approach addresses the challenge of combining arithmetic and comparison operations within a single FV-style scheme. It leverages the native plaintext space $\mathbb{Z}_{p^r}$ for efficient arithmetic while utilizing the digit space $\mathbb{Z}_p$ for comparison evaluation. The reduction and modulus-raising steps that enable this integration are described in detail in the following sections.

\subsection{Reduction from $\mathbb{Z}_{p^r}$ to $\mathbb{Z}_{p}$}
\label{sec:method decomp}
The reduction in FV-style schemes from $\mathbb{Z}_{p^r}$ to $\mathbb{Z}_{p}$ amounts to extracting, for a value $a \in \mathbb{Z}_{p^r}$, its $r$ base-$p$ digits $\{a_i \in \mathbb{Z}_p\}_{i=0}^{r-1}$. This procedure is closely related to the digit-extraction techniques used in BGV/BFV bootstrapping~\cite{gentry2012boot,halevi2021boot,chen2018boot,Geelen2023BootstrappingRevisited,geelen2023boot}. Conceptually, it is equivalent to recursively extracting the least-significant digit (LSD) from the input. In both settings, the central challenge is to extract digits homomorphically. Two widely used approaches are the Halevi-Shoup method~\cite{halevi2021boot} and the Chen-Han method~\cite{chen2018boot}.

Define $a_{i,j}$ for any integer whose least significant digit in base $p$ is $a_i$, with the next $j$ least significant digits all equal to zero. Formally, this means $a_{i,j} = a_i \pmod{p^{j+1}}$. Halevi-Shoup~\cite{halevi2021boot} introduced a polynomial $F_p$ based on the following lemma:

\begin{lemma}[Corollary 5.5 in~\cite{halevi2021boot}]
\label{lm: fe}
For every prime \( p \) and \( e \geq 1 \) there exists a degree-\( p \) polynomial \( F_p \) such that for every integer \( z_0, z_1 \) with \( z_0 \in [p] \) and every \( 1 \leq t \leq e \), we have
\[
F_{p}(z_0 + p^{t} z_1) = z_0 \pmod{p^{t+1}}.
\]
\end{lemma}
We refer to~\cite{halevi2021boot} for details of the explicit construction of $F_p$.

Moreover, the polynomial $F_p$ also satisfies $F_p(a_{i,j}) = a_{i,j+1}$ for every $j \le r$. In other words, this construction enables us to compute a valid $a_{i,j+1}$ from any given $a_{i,j}$ by zeroing out one additional digit at each step. {For example, assuming $a \in \mathbb{Z}_{p^4}$, we have $F_p(a_{i,1}) = F_p(\overline{xy0a_i}) = \overline{x00a_i}$.} We will show how to use this polynomial to obtain digits.

As a toy example, consider $r=4$, so that $a \in \mathbb{Z}_{p^4}$ can be decomposed as $a = \sum_{i=0}^{3} a_i p^i$ with $a_i \in \mathbb{Z}_p$. In Figure~\ref{fig:digit decomp example}, starting with $a = a_{0,0}$ on the first row, we repeatedly compute $a_{0,i+1} = F_p(a_{0,i})$ until we obtain $a_{0,3} = \overline{000a_0}$. This completes the computation for the first row. In the second row, we compute $a_{1,0} = (a - a_{0,1})/p$ and then repeat the same procedure as in the first row. In the third row, we compute $a_{2,0} = \big((a - a_{0,2})/p - a_{1,1}\big)/p$ and again apply the same procedure used in the previous rows. The same process is repeated for the fourth row. At the end, the last entry of each row corresponds to one digit of $a$; specifically, the last element in the $i$-th row gives $a_{i-1}$.

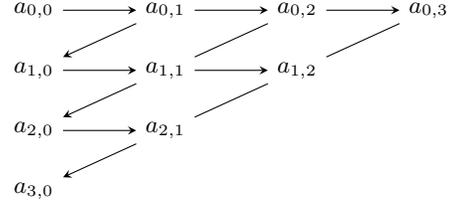
\begin{figure}[htbp]
    \centering
    \begin{tikzpicture}[>=stealth, node distance=0.5cm]
  \node (w00) at (0,0) {\small $a_{0,0}$};
  \node (w01) at (1.75,0) {\small $a_{0,1}$};
  \node (w02) at (3.5,0) {\small $a_{0,2}$};
  \node (w03) at (5.25,0) {\small $a_{0,3}$};

  \node (w10) at (0,-0.8) {\small $a_{1,0}$};
  \node (w11) at (1.75,-0.8) {\small $a_{1,1}$};
  \node (w12) at (3.5,-0.8) {\small $a_{1,2}$};

  \node (w20) at (0,-1.6) {\small $a_{2,0}$};
  \node (w21) at (1.75,-1.6) {\small $a_{2,1}$};

  \node (w30) at (0,-2.4) {\small $a_{3,0}$};

  \draw[->] (w00) -- (w01);
  \draw[->] (w01) -- (w02);
  \draw[->] (w02) -- (w03);

  \draw[->] (w10) -- (w11);
  \draw[->] (w11) -- (w12);

  \draw[->] (w20) -- (w21);

  \draw[->] (w01) -- (w10);
  
  \draw[->] (w11) -- (w20);
  \draw[->] (w21) -- (w30);
  
  \draw[-] (w02) -- (w11);
  \draw[-] (w03) -- (w12);
  \draw[-] (w12) -- (w21);
\end{tikzpicture}

    \caption{Example of digit decomposition in $\mathbb{Z}_{p^4}$.}
    \label{fig:digit decomp example}
\end{figure}

The Chen-Han method~\cite{chen2018boot} uses a similar approach to Halevi-Shoup, with the main difference being in how the last number of each row is computed. In Chen-Han, they introduced another polynomial $G_{p,e}$ based on the following lemma:

\begin{lemma}[Section 3.2 in~\cite{chen2018boot}]
\label{lm:ge}
For every prime \( p \) and \( e \geq 1 \), there exists a polynomial \( G_{p,e} \) of degree at most \((e-1)(p-1)+1\) such that for every integer $z_0$, $z_1$ with $-p/2 < z_0 \le p/2$, we have
\[
G_{p,e}(z_0 + pz_1) = z_0 \pmod{p^e}.
\]
\end{lemma}
We refer to~\cite{chen2018boot} for details of the explicit construction of $G_{p,e}$.

Using Lemma~\ref{lm:ge}, we can apply the polynomial $G_{p,e}$ directly to the input, in this case, the first number of a row, to obtain the last number in the same row. Compared to Halevi-Shoup, this method is more efficient because the degree of $G_{p,e}$ is $(p-1)(e-1)+1$, whereas in Halevi-Shoup it is $p^{e-1}$ due to repeated evaluations of the polynomial $F_p$ over $e-1$ steps.

Another digit-extraction approach is proposed by Geelen et al.~\cite{geelen2023boot}. They use the family $\{G_{p,t}\}_{t\ge 1}$ to compute intermediate values within the same row directly from the input which is the first number in the row. For example, with the first row with first element $a$, they set
\[
a_{0,1} \leftarrow G_{p,2}(a), \quad
a_{0,2} \leftarrow G_{p,3}(a), \quad
a_{0,3} \leftarrow G_{p,4}(a), \ \ldots
\]
so that all subsequent entries share the same input $a$ rather than chaining through $F_p$.

They also give two practical optimizations for evaluating $G_{p,e}$. First, they show a parity property: $G_{p,e}$ is even when $p=2$ and odd when $p$ is odd, so it can be written using only even (resp.\ only odd) exponents. In particular, for odd $p$ one can write $G_{p,e}(x)=x\cdot H(x^2)$ for some polynomial $H$ of degree $(p-1)(e-1)/2$, enabling reuse of a single $x^2$ and an efficient baby-step/giant-step evaluation of $H(x^2)$. Second, they apply a lattice-based coefficient selection to replace $G_{p,e}$ by an equivalent representative with smaller integer coefficients, which reduces ciphertext-noise growth during homomorphic evaluation.

In the context of bootstrapping, where only a few digits need to be removed (to realign noise and scale) rather than all digits, the Chen-Han or Geelen et al.\ methods are advantageous over the Halevi-Shoup method because they avoid computing every intermediate value between $a_{i-1,0}$ and $a_{i-1,r-i+1}$ for each row $i$ and reduce the number of polynomial evaluations. In the context of digit decomposition, however, the situation is different: the first element of each row depends on values from previous rows that lie on the same diagonal, which means that all intermediate numbers must be computed.


In contrast to the previous methods, we propose a digit-decomposition procedure that does not rely on intermediate values. Consider the first row in Figure~\ref{fig:digit decomp example} as an example. Let $b_{i-1}$ denote the first element of row $i$. For the first row, $b_0=a$. To compute $a_{0,3}$, the Halevi-Shoup method requires three evaluations of $F_p$, or alternatively two evaluations of $F_p$ and one evaluation of $G_{p,4}$ using the Chen-Han method; in Geelen et~al., it needs one evaluation each of $G_{p,4}$, $G_{p,3}$, and $G_{p,2}$. In contrast, our method computes $a_{0,3}$ directly as $a_{0,3}=G_{p,4}(b_0)$, avoiding intermediate values and requiring only a single evaluation of $G_{p,4}$. Once $a_{0,3}$ is obtained, the first element of the second row is $b_1=(b_0-a_{0,3})/p$. The same process is then applied to the remaining rows, giving $b_2=(b_1-a_{1,2})/p$ and $b_3=(b_2-a_{2,1})/p$. An illustration of this procedure is provided in Figure~\ref{fig:digit decomp prop example}. As before, the last element in the $i$-th row gives $a_{i-1}$, where $\{a_i\}_{i=0}^{r-1}$ are the digits of $a$.  {Note that, except for the first row, the first numbers of the corresponding rows in
Figure~\ref{fig:digit decomp example} and Figure~\ref{fig:digit decomp prop example}
come from different constructions and therefore need not coincide numerically.
However, in both constructions, the last element in each row yields the same digit output.}

\begin{figure}[htbp]
    \centering
    \begin{tikzpicture}[>=stealth, node distance=0.5cm]
  \node (w00) at (0,0) {\small $b_0=a$};
  \node (w01) at (1.75,0) [text=lightgray] {\small$a_{0,1}$};
  \node (w02) at (3.5,0) [text=lightgray] {\small$a_{0,2}$};
  \node (w03) at (5.25,0) {$a_{0,3}$};

  \node (w10) at (0,-0.8) {\small$b_{1}$};
  \node (w11) at (1.75,-0.8) [text=lightgray] {\small$a_{1,1}$};
  \node (w12) at (3.5,-0.8) {\small$a_{1,2}$};

  \node (w20) at (0,-1.6) {\small$b_{2}$};
  \node (w21) at (1.75,-1.6) {\small$a_{2,1}$};

  \node (w30) at (0,-2.4) {\small$b_{3}$};

  \draw[-] (w00) -- (w01);
  \draw[-] (w01) -- (w02);
  \draw[->] (w02) -- (w03);


  \draw[-] (w10) -- (w11);
  \draw[->] (w11) -- (w12);


  \draw[->] (w20) -- (w21);

  \draw[->] (w03) -- (w10);
  
  \draw[->] (w12) -- (w20);
  \draw[->] (w21) -- (w30);
  
\end{tikzpicture}

    \caption{Example of proposed digit decomposition in $\mathbb{Z}_{p^4}$.}
    \label{fig:digit decomp prop example}
\end{figure}
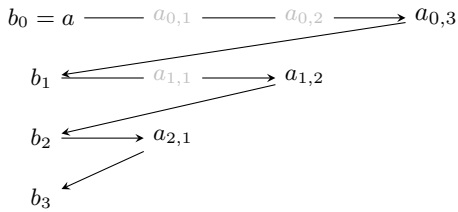

As illustrated in Figure~\ref{fig:digit decomp example}, when $r=4$, the Halevi-Shoup method requires a total of six evaluations of the polynomial $F_p$. If the Chen-Han method is used, the cost becomes three evaluations of $F_p$ together with one evaluation each of $G_{p,4}$, $G_{p,3}$, and $G_{p,2}$. If the Geelen et~al.\ method is used, the cost becomes three evaluations of $G_{p,2}$, two evaluations of $G_{p,3}$, and one evaluation of $G_{p,4}$. In comparison, our method avoids all intermediate computations by directly applying a single $G_{p,e}$ polynomial at each step. This results in a substantial reduction in the number of evaluations. For $r=4$, this saves three evaluations of $F_p$ in total.

In the general case, the Halevi-Shoup method requires $O(r^2)$ evaluations of $F_p$, the Chen-Han method requires $O(r)$ evaluations of $F_p$ plus $O(r)$ evaluations of $G_{p,e}$, and the Geelen et~al.\ method requires $O(r^2)$ evaluations of $G_{p,e}$, {since all intermediate values are required for subsequent rows}. By contrast, our approach reduces the cost to exactly $r$ evaluations of $G_{p,e}$, one per row. Thus, our method achieves lower computational overhead and improved efficiency compared to Halevi-Shoup, Chen-Han, and Geelen et~al., although it does not reduce the multiplicative depth relative to those methods.

\subsubsection*{Switch to the digit space $\mathbb{Z}_p$}
The digits obtained from the above decomposition process have different plaintext moduli, because the plaintext modulus is reduced only during the division-by-$p$ operation in the vertical direction. We need to unify the plaintext space of these digits to $\mathbb{Z}_p$ before evaluating comparison operations. For example, in Figure~\ref{fig:digit decomp example} and Figure~\ref{fig:digit decomp prop example}, $a_{0,3}$ represents the LSD of the input $a$ and has modulus $p^4$. The next three digits, $a_{1,2}$, $a_{2,1}$, and $b_3=a_{3,0}$, have moduli $p^3$, $p^2$, and $p$, respectively. {More formally, at the message level, the division-by-$p$ operation is defined as}
\[
\mathsf{DivideByP}(kp) = k \pmod{p^{t-1}},
\]
{for some $k \in \mathbb{Z}_{p^{t-1}}$ and $t>1$. In the actual procedure, the input is a ciphertext encrypting a value congruent to $kp \pmod{p^t}$, and the output is a ciphertext encrypting $k \pmod{p^{t-1}}$}. Thus, both the message and the plaintext modulus are reduced by a factor of $p$~\cite{chen2018boot,halevi2021boot}. This homomorphic division-by-$p$ is realized by multiplying the ciphertext by the constant $p^{-1} \pmod{Q}$~\cite{chen2018boot,halevi2021boot}.

Now that the digits have different plaintext moduli, we need to reduce the plaintext modulus without reducing the message. Therefore, we define a function $\mathsf{ChangeModToP}$ as:
\[
\mathsf{ChangeModToP}(k) = k \pmod{p},
\]
which takes as input an encryption of $k \pmod{p^t}$ and returns $k \pmod{p}$, for some $k \in \mathbb{Z}_{p^{t}}$. In other words, it sets the plaintext modulus of the ciphertext from $p^t$ to $p$ without changing the content of the ciphertext. The correctness of $\mathsf{ChangeModToP}$ follows from Lemma~\ref{lm: mod}.
\begin{lemma}
\label{lm: mod}
For every prime \( p \) and integers \(1\leq m\leq n\), if $a\equiv b \pmod{p^n}$, then $a \equiv b \pmod{p^m}$.
\end{lemma}
\begin{proof}
Since $a\equiv b \pmod{p^n}$, we have that $p^n$ divides $a-b$. Because $m\le n$, it follows that $p^m$ divides $p^n$, hence $p^m$ divides $a-b$ and therefore $a \equiv b \pmod{p^m}$.
\end{proof}
In other words, Lemma~\ref{lm: mod} ensures that if we have the correct base-$p$ digits modulo $p^t$ for any $t \geq 1$, then those digits are also correct modulo $p$. Operationally, this does not change the ciphertext itself; it only affects the decryption step in Algorithm~\ref{alg:decryption_algo}, particularly line~2, where the final reduction is taken modulo the current plaintext modulus $p$ rather than $p^r$. This is consistent with the BGV relation
\[
\mathbf{c}_0 + \mathbf{c}_1 \cdot \mathbf{s} \;\equiv\; \mathbf{m} + p^{r} \mathbf{e} = \mathbf{m} + p\cdot \left(p^{r-1}\right) \mathbf{e} \pmod{Q}.
\]
from Eq.~\eqref{eq: bgv}. With Lemma~\ref{lm: mod} and $\mathsf{ChangeModToP}$, we can reduce the plaintext modulus of all digits to $p$.

For example, given an input $a \in \mathbb{Z}_{p^4}$, we know its LSD is $a_0 \in \mathbb{Z}_{p}$. As shown in Figure~\ref{fig:digit decomp example} and Figure~\ref{fig:digit decomp prop example}, by evaluating either the polynomial $F_p$ or $G_{p,e}$, we can obtain $a_{0,3} = a_0 \pmod{p^4}$. By Lemma~\ref{lm: mod}, we know $a_{0,3} = a_0 \pmod{p^m}$ for any $1 \leq m \leq 4$. Thus, $a_{0,3} = a_0 \pmod{p}$. Accordingly, we can obtain the digit $a_0$ modulo $p$ by directly reducing the modulus of $a_{0,3}$ to $p$. This is achieved by calling $\mathsf{ChangeModToP}(a_{0,3})$. Similarly, we have $a_1 = \mathsf{ChangeModToP}(a_{1,2})$ and $a_2 = \mathsf{ChangeModToP}(a_{2,1})$. By integrating the digit decomposition with $\mathsf{ChangeModToP}$, we obtain the reduction procedure from FV over $\mathbb{Z}_{p^r}$ to $\mathbb{Z}_{p}$, as shown in Algorithm~\ref{alg:reduction}.




 

\begin{algorithm}[t]
\caption{Reduction from $\mathbb{Z}_{p^r}$ to $\mathbb{Z}_{p}$}
\label{alg:reduction}
\SetAlgoLined
\LinesNumbered
\KwIn{$a \in \mathbb{Z}_{p^r}$}
\KwOut{$\{a_i\}_{i=0}^{r-1}$ with $a_i \in \mathbb{Z}_p$ and $a = \sum_{i=0}^{r-1} a_i p^i$}
$b \leftarrow a$\;
\For{$i \gets 0$ \KwTo $r-2$}{
    $a_i \leftarrow G_{p,\,r-i}(b)$\tcp*{Horizontal reduction}
    $b \leftarrow \mathsf{DivideByP}\!\big(b - a_i\big)$\tcp*{Vertical reduction}
}
$a_{r-1} \leftarrow b$\;
\For{$i \gets 0$ \KwTo $r-1$}{
    $a_i \leftarrow \mathsf{ChangeModToP}(a_i)$\tcp*{Modulus reduction to $\mathbb{Z}_p$}
}
\Return $\{a_i\}_{i=0}^{r-1}$\;
\end{algorithm}

\subsubsection*{Complexity of the reduction function}
As described in Algorithm~\ref{alg:reduction}, the reduction from $\mathbb{Z}_{p^r}$ to $\mathbb{Z}_p$ consists of three steps: horizontal reduction, vertical reduction, and modulus reduction. The vertical reduction involves only subtraction and division by $p$, where in FV-style HE the division by $p$ is implemented as a scalar multiplication. Since subtraction and scalar multiplication incur negligible cost~\cite{chen2018boot}, their overhead is marginal. The modulus reduction step operates only on the plaintext modulus of the message and does not require any ciphertext evaluation, and can therefore be considered essentially free. Consequently, the dominant cost of the reduction function lies in the horizontal reduction.  

The horizontal reduction requires evaluating the polynomial $G_{p,r-i}(\cdot)$. Lemma~\ref{lm:ge} guarantees that such a polynomial always exists with degree $(p-1)(e-1)+1$. Moreover, $G_{p,e}(x)$ can be expressed as $x \cdot H(x^2)$ or $H(x^2)$, where the degree of $H$ is half the degree of $G_{p,e}$. This implies that evaluating $G_{p,e}$ requires $O\!\left(\sqrt{(p-1)(e-1)+1}\right)$ multiplications using the Paterson-Stockmeyer algorithm, and consumes $O\!\left(\log((p-1)(e-1)+1)\right)$ multiplicative depth.  

In Algorithm~\ref{alg:reduction}, reducing the plaintext space of $a \in \mathbb{Z}_{p^r}$ to its digit representation $\{a_i\}_{i=0}^{r-1} \in \mathbb{Z}_p$ involves $r$ rows in the horizontal direction. The last element in each row, $a_{i,r-i-1}$, corresponds to the $i$-th digit of the input $a$. In each row, the polynomial $G_{p,r-i}$ is evaluated once to compute $a_{i,r-1}$. Thus, the reduction requires exactly one evaluation of each $G_{p,r-i}$, for a total of $\sum_{i=0}^{r-1} \sqrt{(p-1)(r-i-1)+1}$ multiplications.  

The multiplicative depth of the reduction is determined by the longest computation path, which corresponds to traversing all digits as illustrated in Figure~\ref{fig:digit decomp prop example}. Along this path, there are $r-1$ evaluations of $G_{p,e}$, yielding a depth of $\sum_{i=0}^{r-1} \log\!\big((p-1)(r-i-1)+1\big)$.

\subsection{Efficient Comparison in $\mathbb{Z}_{p}$}
\label{sec:method befv}

Given the reduction from $\mathbb{Z}_{p^r}$ to $\mathbb{Z}_p$ (Section~\ref{sec:method decomp}), we evaluate comparisons by applying standard per-digit tests in the digit space $\mathbb{Z}_p$ and then aggregating the digit results. Unlike prior work that centers comparisons inside special plaintext encodings, we use interpolation polynomials only as \emph{per-digit predicates} in $\mathbb{Z}_p$, while keeping arithmetic in $\mathbb{Z}_{p^r}$ and avoiding special encodings.

Typical comparison functions include equality (EQ), inequality (NEQ), less-than (LT), less-than-or-equal (LE), greater-than (GT), and greater-than-or-equal (GE). In this work, however, we focus on constructing only EQ and LT homomorphically, since GT can be derived in the same manner as LT, while NEQ, GE, and LE can be obtained from simple transformations: $\text{NEQ} = 1 - \text{EQ}$, $\text{GE} = 1 - \text{LT}$, and $\text{LE} = 1 - \text{GT}$. Consistent with~\cite{tan2020efficientsefv,iliashenko2021fasterbgvsefv}, we rely on univariate interpolation, which is faster than bivariate in homomorphic settings. To make this possible, we restrict comparisons to zero, since arbitrary comparisons such as $a < b$ can always be reduced to checking whether $a-b < 0$.

\subsubsection*{Digit-wise comparison}
The EQ function is defined as $EQ(x,0)=1$ for $x=0$ and $EQ(x,0)=0$ otherwise. By Theorem~\ref{lm: fermat}, we can express it as $F_{EQ}(x,0)=1-x^{p-1}$. This follows directly from Fermat’s little theorem and makes EQ straightforward and efficient.

In contrast, constructing a polynomial for LT is not as trivial as in the case of EQ. The LT function is defined as $LT(x, 0) = 1$ for $x < 0$ and $LT(x, 0) = 0$ otherwise. The polynomial for LT can be derived using a univariate interpolation polynomial over $\mathbb{Z}_p$~\cite{iliashenko2021fasterbgvsefv}, resulting in
\[F_{LT}(x,0) = \frac{p+1}{2} x^{p-1} + \sum_{\substack{i=1,\text{odd}}}^{p-2} c_i x^i,\]
where $c_i = \sum_{j=1}^{(p-1)/2} j^{p-1-i}$. Similarly, the polynomial $F_{GT}(x)$ for the GT function can be derived. Unlike EQ, the LT construction requires higher-degree interpolation and more terms, making it computationally more expensive. Both $F_{EQ}$ and $F_{LT}$ have the same degree of $p-1$. We note that both $F_{EQ}$ and $F_{LT}$ are defined over $\mathbb{Z}_p$.

\subsubsection*{Digit-wise evaluation and aggregation}
The polynomials $F_{LT}$ and $F_{EQ}$ defined in the previous paragraph are over $\mathbb{Z}_p$ and therefore cannot be applied directly to inputs in $\mathbb{Z}_{p^r}$. Suppose we want to compare a number $z \in \mathbb{Z}_{p^r}$ with $0$. We first apply the reduction function to map $z$ to its digit representation $\{z_i\}_{i=0}^{r-1} \in \mathbb{Z}_p$. Since each digit now lies in $\mathbb{Z}_p$, we can evaluate $F_{LT}$ or $F_{EQ}$ on each digit individually. By aggregating the per-digit results, we obtain the comparison result for the original number $z$. The aggregation can be defined as
\[
\prod_{i=0}^{r-1} F_{EQ}(z_i,0)
\]
for the EQ case, or
\[
\sum_{i=0}^{r-1} F_{LT}(z_i,0)\,\prod_{j=i+1}^{r-1} F_{EQ}(z_j,0)
\]
for the LT case. The result of the aggregation step lies in $\mathbb{Z}_p$.

\subsubsection*{Comparing two numbers in $\mathbb{Z}_{p^r}$}
Suppose we have two integers $a,b \in \mathbb{Z}_{p^r}$ and want to test whether $a<b$. Section~\ref{sec:method limit logic fv} showed that comparison directly over $\mathbb{Z}_{p^r}$ is infeasible for $r>1$. To overcome this, we first apply the reduction process, allowing us to use polynomials defined over the digit space $\mathbb{Z}_p$. Specifically, we compute the difference $d=a-b$ in $\mathbb{Z}_{p^r}$, reduce $d$ into its base-$p$ digits in $\mathbb{Z}_p$, and then compare each digit to zero using $F_{LT}$ and $F_{EQ}$. Finally, the per-digit results are aggregated to obtain the overall comparison between $a$ and $b$. The full procedure is summarized in Algorithm~\ref{alg:compare ab} for the LT operation.

\begin{algorithm}[t]
\SetAlgoLined
\LinesNumbered
\caption{LT Comparison in $\mathbb{Z}_{p^r}$}
\label{alg:compare ab}
\KwIn{integers $a \bmod p^r$ and $b \bmod p^r$}
\KwOut{comparison result in $\mathbb{Z}_p$}

$d \leftarrow a - b$ \;
$\textit{digits} \leftarrow \mathsf{Reduce}(d, p, r)$ \;

\For{$i \gets 0$ \KwTo $r-1$}{
    $LT_i \leftarrow F_{\mathrm{LT}}(\textit{digits}_i, 0)$ \tcp*{Digit-wise comparison}
    $EQ_i \leftarrow F_{\mathrm{EQ}}(\textit{digits}_i, 0)$\;
}

$LT \leftarrow LT_{r-1}$ \;
$EQ \leftarrow EQ_{r-1}$ \;

\For{$i \gets r-2$ \KwTo $0$}{
    $LT \leftarrow LT + EQ \cdot LT_i$\tcp*{Aggregation}
    $EQ \leftarrow EQ \cdot EQ_i$\;
}

\Return $LT$\;
\end{algorithm}

The comparison computation in $\mathbb{Z}_p$ is dominated by the evaluation of interpolation polynomials of degree $p-1$. For $r$ digits, this costs $r\sqrt{p-1}$ non-scalar multiplications and consumes $\log p$ multiplicative depth. The aggregation step uses only $O(r)$ additions and multiplications, which can be performed efficiently.

\subsection{Modulus Raising from $\mathbb{Z}_p$ to $\mathbb{Z}_{p^r}$}
\label{sec:method lift}

Let $b$ denote the result after the aggregation step (Section~\ref{sec:method befv}). Since $b \in \mathbb{Z}_p$, the output of the comparison computation still resides in the $\mathbb{Z}_p$ domain. To enable further arithmetic operations, we propose a function that homomorphically maps $b \bmod p$ in $\mathbb{Z}_p$ back to $b \bmod p^r$ in FV over $\mathbb{Z}_{p^r}$, which we call \emph{modulus raising}.

{To perform this modulus raising, the first step is to enlarge the plaintext modulus of the underlying message. Given an encryption of $k \pmod{p}$, this step produces $\bar{k} \pmod{p^r}$, where only the modulus has been extended to $p^r$. However, in general $\bar{k} \neq k \pmod{p^r}$. Simply enlarging the modulus ensures that the least significant digit (LSD) of $\bar{k}$ matches $k$, while the higher-order digits remain undetermined. In other words, the reverse of Lemma~\ref{lm: mod} does not hold: although $\bar{k}$ is congruent to $k$ modulo $p$, the equality does not extend to modulo $p^r$.}

{To recover the correct value $k \pmod{p^r}$, we must clear the higher digits of $\bar{k}$. This task is equivalent to extracting the LSD of $\bar{k}$, which can be accomplished through horizontal reduction in Algorithm~\ref{alg:reduction}, restricted to the first row. More directly, we define
\[
\mathsf{RaiseModToP2R}(k) = G_{p,r}(\bar{k}) \pmod{p^r},
\]
which can be viewed as the inverse of $\mathsf{ChangeModToP}$. Thus, $\mathsf{RaiseModToP2R}$ internally performs modulus extension followed by the evaluation of $G_{p,r}$, and directly outputs the canonical lifted value in $\mathbb{Z}_{p^r}$:
\[
k = G_{p,r}(\bar{k}), \quad k \in \mathbb{Z}_{p^r}.
\]}
This ensures that the resulting value is correctly represented in the FV domain over $\mathbb{Z}_{p^r}$.

\subsection{Overall Complexity}
\label{sec:method complexity}
Among all available operations in FHE, non-scalar (ciphertext-ciphertext) multiplication is the costliest primitive. Therefore, it is essential to reduce the number of non-scalar multiplications.

The LT evaluation for direct comparison in $\mathbb{Z}_{p'}$ where $p' \approx p^r$ requires $O(\sqrt{p'}) \approx O(\sqrt{p^{r}}) = O(p^{r/2})$ non-scalar multiplications. Our \emph{space switching} instead begins with digit extraction: extracting the $i$-th digit evaluates a polynomial of degree $(p-1)(r-i-1)+1$, costing $O\left(\sqrt{(p-1)(r-i-1)+1}\right) \approx O(\sqrt{pr})$ multiplications, so extracting all $r$ digits costs about $O(r\sqrt{pr})$. Per-digit comparison then costs $O(\sqrt{p})$ each, for a total of $O(r\sqrt{p})$. Aggregating the per-digit results uses $O(r)$ non-scalar multiplications. Finally, the modulus-raising step costs $O\left(\sqrt{(p-1)(r-1)+1}\right) \approx O(\sqrt{pr})$. Summing these terms yields
\[
O\left(r\sqrt{pr} + r\sqrt{p} + r + \sqrt{pr}\right) \approx O\!\left(r\sqrt{pr}\right),
\]
so \emph{space switching} reduces the non-scalar multiplication count from $O\!\left(p^{r/2}\right)$ to $O\!\left(r\sqrt{pr}\right)$.

Let $w = r \log_2 p$ be the bit-width. As $w$ increases, direct comparison costs $O(p^{r/2})$. This grows exponentially in $r$ when $p$ is fixed, and as $p^{r/2}$ when $r$ is fixed. In contrast, \emph{space switching} costs $O(r^{3/2}\sqrt{p})$ up to constants. For small fixed $p$ (e.g., $p \le 17$), the direct comparison cost explodes with $r$, while our \emph{space switching} cost grows roughly linearly in $r$ with an additional $\sqrt{p}$ factor. For fixed $r$, increasing $p$ raises our cost by $\sqrt{p}$, whereas the direct method grows as $p^{r/2}$. This gap explains why \emph{space switching} dominates at larger bit-widths.

\section{Experiments}
\label{s:expt}

In this section, we describe the experimental setup and evaluation methodology. Section~\ref{sec:expt setting} presents the platforms, hardware environment, and software dependencies. Section~\ref{sec:expt workloads} outlines the workloads used in our evaluation, and Section~\ref{sec:expt schemes} introduces the schemes against which we compare.

\subsection{Implementation Details}
\label{sec:expt setting}

\textbf{System Setup.} We implement our work on HElib~\cite{halevi2014helib} and the NTL library~\cite{shoup2001ntl}. We mainly leverage the BGV scheme~\cite{brakerski2014bgv} as the underlying HE scheme. We also adapted the implementation of~\cite{geelen2023boot} for efficient polynomial evaluation. All experiments are conducted on a server with Intel(R) Xeon(R) W-2235 equipped with 128GB RAM running at 3.80GHz. All parameters for BGV are set to ensure $\lambda > 120$, i.e., at least $120$-bits security unless otherwise specified. The security level can be approximated by the LWE estimator~\cite{albrecht2015lwe}. {All experiments are single-threaded, consistent with~\cite{geelen2023boot,iliashenko2021fasterbgvsefv}.} The code can be accessed in\textcolor{blue}{~\href{https://github.com/UCF-Lou-Lab-PET/Universal-BGV}{this repository}}.


\subsubsection*{Parameter Setup}
We set the encryption parameters for all experiments according to the following strategies.

{\em Fixed parameters}. We set the Hamming weight of the secret key as $t=120$ for all experiments. The Hamming weight represents the number of non-zero coefficients in the secret key. We set the columns of the key-switching matrix to $c=2$. We align these two parameters with previous works~\cite{tan2020efficientsefv, iliashenko2021fasterbgvsefv, morimura2023accelerating}.
    
    {\em The plaintext modulus}. The plaintext modulus in the FV space is determined by $p^r$. We can choose $p$ and $r$ according to the input bit-width $b$ such that $p^r$ has a size of at least $b$ bits. The choice of $p$ and $r$ determines both the complexity of the plaintext switching and the complexity of comparison operations as it affects the polynomials $F_p$ and $G_{p,r}$.

    {\em The ciphertext modulus}. The ciphertext modulus $Q$ affects the multiplicative depth, that is the maximum number of non-scalar multiplications an HE instance can support before decryption fails due to noise~\cite{halevi2021boot,tan2020efficientsefv,chen2017seal}. Therefore, $Q$ is chosen to match the depth of the evaluation circuit.
    
    {\em The degree of polynomial ring}. Once we have chosen $p$, $r$ and $Q$, we can choose an appropriate degree of polynomial to meet the security level $\lambda > 120$. By setting the cyclotomic order of the polynomial ring $m$, at the same time we also set the degree of polynomial ring $n$, the order of the base prime $d = \textsf{Ord}(m,p)$, and the number of SIMD slots $\ell=n/d$.

    

\subsection{Workloads Evaluated}
\label{sec:expt workloads}


We evaluate the performance of our proposed method across four aspects. First, we analyze the runtime of comparison operations, focusing on the LT function over inputs of different bit-widths. Second, we assess its effectiveness on database workloads using the TPC-H benchmark suite~\cite{tpc2022tpcbenchmarks}, specifically Query~6. This query involves four columns, each encrypted into a single ciphertext, resulting in four ciphertexts in total. All columns are represented as 16-bit integers. The query includes two core operations: \textit{filtering}, which applies comparison conditions to select rows, and \textit{aggregation}, which computes summary values such as \texttt{SUM} or \texttt{COUNT} over the filtered results. To ensure compatibility with our encrypted framework, we preprocess the dataset so that all values are represented as integers, and we select parameters so that every ciphertext can pack all rows (i.e., the number of slots exceeds the dataset size). Third, we provide a runtime breakdown to identify the relative cost of each sub-operation in our method. Finally, we perform a sensitivity study by examining LT performance across varying bit-widths.

\subsection{Schemes Evaluated}
\label{sec:expt schemes}
We compare three approaches that support both arithmetic and comparison operations on encrypted data, all producing exact results without approximation errors: space switching (our work), scheme switching, and the work from~\cite{morimura2023accelerating} which we refer to as \textit{direct comparison}. We exclude schemes that do not give exact results, such as polynomial approximation, because the workloads that we test require precise comparison results.

Space switching is designed to support both arithmetic operations over $\mathbb{Z}_{p^r}$ and comparison operations natively within the same scheme, avoiding the need to alternate between arithmetic- and bit-oriented schemes. By contrast, scheme switching~\cite{boura2020chimera,lu2021pegasus,bian2023he3db} combines a word-wise HE scheme with a bit-wise scheme, incurring overhead from ciphertext conversions.


Direct comparison~\cite{morimura2023accelerating} evaluates both arithmetic and comparisons in a single arithmetic HE scheme over $\mathbb{Z}_p$. This avoids conversions but requires $p$ to grow exponentially with bit-width; as $p$ increases, the interpolation degree ($p{-}1$) inflates multiplicative depth, noise, and evaluation cost.



\section{Evaluation Results}
\label{s:results}


\subsection{Performance on the LT Operator}
\label{sec:expt comp perf}
We compare our \emph{space switching} against direct comparison~\cite{morimura2023accelerating} and HE$^3$DB~\cite{bian2023he3db}, which we adopt as the representative scheme-switching method. Figure~\ref{fig:expt-LT} presents the runtime of space switching in comparison with scheme switching and direct comparison for the LT operator across 8- to 20-bit inputs. The parameters for space switching and direct comparison are shown in Table~\ref{tab:param ss and dc}, while parameters for scheme switching follow~\cite{bian2023he3db}. Runtimes are reported as amortized per-slot runtime.

\begin{table}[h]
    \centering
    \caption{Parameters for space switching and direct comparison.}
    \label{tab:param ss and dc}
    \begin{tblr}{
      colspec = {|c|c|c|c|},
      row{1} = {font=\fontsize{8}{11},rowsep=1pt},
      row{2-Z} = {font=\fontsize{8}{11},rowsep=1pt},
    }
    \hline
    \textbf{Bit-} & \SetCell[r=2]{c}\textbf{$\log Q$} & \textbf{Space Switching} & \textbf{Direct Comparison} \\
    \textbf{width} & & $(p^r,m,\ell)$ & $(p,m,\ell)$ \\\hline\hline
       8  & 376  & $(7^3,17195,1080)$ & $(257,30256,960)$ \\
       12 & 672  & $(7^5,30025,3000)$ & $(4099,34833,2728)$ \\
       16 & 840  & $(41^3,59595,3808)$ & $(65537,50643,3960)$ \\
       20 & 1280 & $(37^4,46981,8540)$ & $(1048583,39872,9968)$ \\\hline
    \end{tblr}
\end{table}

\begin{figure}[h]
    \centering
    \includegraphics[width=\linewidth]{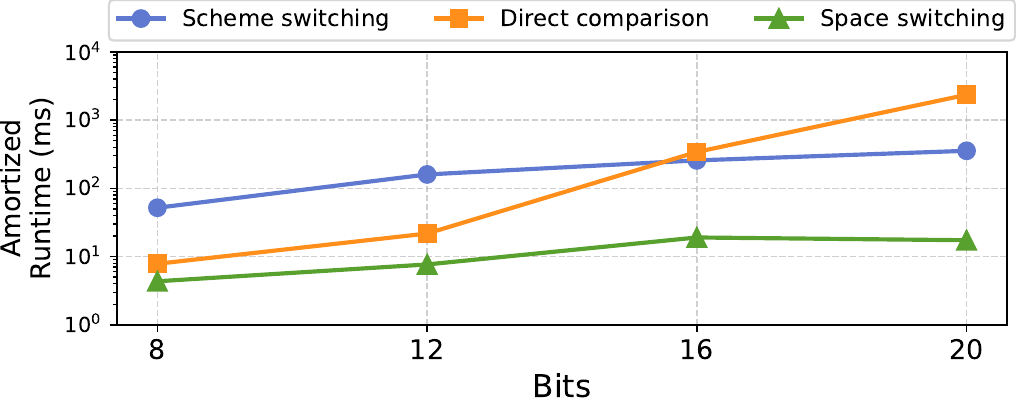}
    \caption{Amortized runtime comparison for the LT operator across Scheme Switching, Direct Comparison, and Space Switching.}
    \label{fig:expt-LT}
\end{figure}



For the LT operator (Figure~\ref{fig:expt-LT}), space switching requires 4.34 ms and 7.66 ms for 8- and 12-bit inputs, respectively. These values represent up to a 3$\times$ speedup over direct comparison and a 21$\times$ speedup over scheme switching. In this lower bit-width range, direct comparison runs faster than scheme switching, but both remain significantly slower than space switching. As the bit-width increases to 16 and 20, the trend reverses: direct comparison incurs very high cost, exceeding even scheme switching, while space switching maintains runtimes of 19.05 ms and 17.36 ms, corresponding to 18$\times$-136$\times$ improvements over direct comparison and 13$\times$-20$\times$ over scheme switching. This indicates that the relative efficiency of direct comparison is limited to small bit-widths, whereas space switching remains the most scalable approach, consistent with the complexity analysis in Section~\ref{sec:method complexity}.



\subsection{Database Workload Evaluation Results}
\label{sec:expt db results}

For the database workload, we compare space switching with two alternative methods, direct comparison and scheme switching, in the database workload with varying numbers of rows. The results are shown in Figure~\ref{fig:expt-DB-Q6}.

\begin{figure}[h]
\centering
\includegraphics[width=\linewidth]{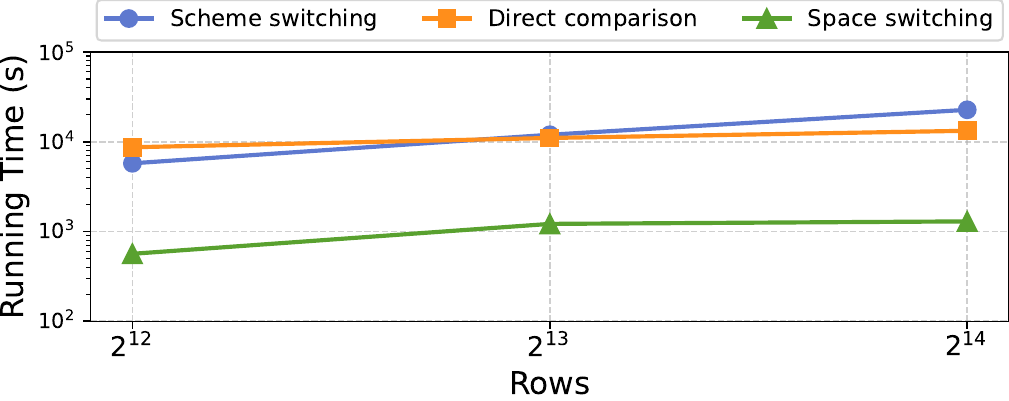}
\caption{Runtime comparison for the TPC-H Benchmark~\cite{tpc2022tpcbenchmarks} Query 6 across scheme switching, direct comparison, and space switching.}
\label{fig:expt-DB-Q6}
\end{figure}


\begin{table*}[t]
\centering
\caption{The runtime breakdown of the LT evaluation. Time is in seconds.}
\setlength{\tabcolsep}{3pt}            
\renewcommand{\arraystretch}{1}     
\begin{tblr}{
  colspec = {c c c c c c c},
  row{1} = {font=\bfseries},
  row{2-Z} = {rowsep=1pt},
}
\hline
\textbf{Bit-width} & \textbf{(p, r)} & \textbf{Reduction} & \textbf{Comparison} & \textbf{Aggregation} & \textbf{Modulus Raising} & \textbf{Total} \\
\hline\hline
\textbf{8}  &(7, 3)  & 1.89 & 1.65 & 0.44 & 0.71 & 4.69\\
\textbf{12} &(7, 5)  & 12.06 & 6.45 & 2.21 & 2.27 & 22.99\\
\textbf{16} &(41, 3) & 30.32 & 29.04 & 2.51 & 10.69 & 72.56\\
\textbf{20} &(37, 4) & 71.65 & 54.6 & 5.77 & 16.26 & 148.28\\
\hline
\end{tblr}%
\label{tab:breakdown}
\end{table*}

\begin{table*}[htbp]
\centering
\caption{Runtime of the proposed method for LT evaluation across different bit-widths.}

\label{tab:comp perf}
\begin{tblr}{
    colspec = {|c | c | c | c | c | c| c|},
    row{1} = {font=\bfseries},
    row{2-Z} = {rowsep=1pt},
    }
\hline
\textbf{Bit-width} & \textbf{(p, r)} & \textbf{$\log$ Q} & \textbf{(m, n)} & $\mathbf{\ell}$ & \textbf{Total Running Time (s)} & \textbf{Amortized Time (ms)}\\
\hline\hline
\SetCell[r=3]{c}\textbf{8} & (5, 4) & 456 & (12851, 12600) & 840 & 4.68 & 5.57 \\
& (7, 3) & 376 &(17195, 12960) & 1080 & 4.69 & 4.34 \\
& (17, 2) & 432 & (21513, 14000) & 1400 & 6.30 & 4.50 \\
\hline
\SetCell[r=4]{c}\textbf{12} & (7, 5) & 672 & (30025, 24000) & 3000 & 22.99 & 7.66  \\
& (11, 4) & 784 & (44289, 23328) & 3888 & 40.37 & 10.38 \\
& (17, 3) & 672 & (55328, 20736) & 1152 & 38.43 & 33.36 \\
& (67, 2) & 780 & (32259, 21504) & 2688 & 33.26 & 12.37 \\
\hline
\SetCell[r=3]{c}\textbf{16} & (17, 4) & 872 & (56889, 31752) & 756 & 68.98 & 91.24 \\
& (41, 3) & 840 & (59595, 30464) & 3808 & 72.56 & 19.05 \\
& (257, 2) & 896 & (40025, 32000) & 2000  & 79.82 & 39.91 \\
\hline
\SetCell[r=3]{c}\textbf{20} & (17, 5) & 1200 & (83811, 44064) & 7344 & 185.42 & 25.25 \\
& (37, 4) & 1280 & (46981, 42700) & 8540 & 148.28 & 17.36 \\
& (131, 3) & 1280 & (62447, 48600) & 8100 & 223.35 & 27.57 \\
\hline
\end{tblr}
\end{table*}



Direct comparison suffers very high runtime because the interpolation degree grows exponentially with input size: $2^{12}$ rows take 2 hours 24 minutes and $2^{14}$ exceed 3.5 hours. Although it avoids reduction and digit extraction, the bottleneck is evaluating the comparison polynomial of degree $p-1$; for a 16-bit prime $p$ this becomes prohibitively expensive.

Scheme switching performs somewhat better than direct comparison at smaller scales, with runtimes of 1 hour 35 minutes and 3 hours 19 minutes for $2^{12}$ and $2^{13}$ rows, respectively. However, at $2^{14}$ rows, scheme switching becomes significantly worse, taking more than 6 hours, nearly twice as slow as direct comparison. This inefficiency arises from the way comparison operations are handled. Initially, the database consists of $N$ ciphertexts in LWE form, and to evaluate a comparison predicate each ciphertext must be processed individually. Every LWE ciphertext is first converted into a TFHE ciphertext, where the comparison function is executed. After all rows have been processed, the comparison results are repacked into a single RLWE ciphertext to support aggregation queries. This pipeline introduces several layers of overhead: (i) repeated conversions from LWE to TFHE for every ciphertext; (ii) row-by-row evaluation in TFHE, where aggregating multiple comparison outputs through homomorphic AND operations requires repeated bootstrapping-heavy computations; and (iii) repacking from many LWE ciphertexts back into a single RLWE ciphertext, which adds costly external products and key switches. These costs compound with dataset size, explaining why scheme switching scales poorly and even becomes slower than direct comparison at $2^{14}$ rows.

Space switching, by contrast, achieves substantially lower runtime across all tested scales. For $2^{12}$ rows, space switching completes in only 9 minutes 24 seconds, yielding a 15$\times$ speedup over direct comparison and a 10$\times$ speedup over scheme switching. At $2^{13}$ rows, space switching runs in 20 minutes, corresponding to 9$\times$ and 10$\times$ improvements, respectively. For the largest dataset of $2^{14}$ rows, space switching completes in 21.5 minutes, giving a 10$\times$ speedup compared to direct comparison and a 17$\times$ improvement over scheme switching.

Overall, space switching consistently reduces runtime by more than an order of magnitude relative to direct comparison and up to 17$\times$ relative to scheme switching. More importantly, while both baselines show steep runtime growth as the dataset scales, with direct comparison hindered by the cost of evaluating high-degree polynomials and scheme switching burdened by scheme-switching and packing overheads, space switching maintains near-linear growth and remains within practical runtime ranges even for the largest tested dataset. This scalability advantage underscores the suitability of space switching for encrypted query processing at larger scales.


In addition to its runtime advantages, space switching also offers a more favorable memory footprint than scheme switching. Across datasets of size $2^{12}$, $2^{13}$, and $2^{14}$ rows, space switching consumes 6.6 GB, 12.6 GB, and 44 GB of memory, respectively, whereas scheme switching requires 57.5 GB, 62.2 GB, and 62.3 GB. Direct interpolation uses 5.8 GB, 13.6 GB, and 31 GB. A similar trend appears in evaluation key size. Space switching requires 1.3 GB, 2.5 GB, and 8.5 GB of evaluation keys, while scheme switching consistently requires 26.7 GB across all cases. Direct interpolation requires 0.6 GB, 1.9 GB, and 5.3 GB. These results show that space switching achieves substantial runtime improvements without incurring the very large memory and key-size overheads of scheme switching, while remaining within a similar overall resource range as direct interpolation.

\subsubsection*{Larger bit-width}
We also evaluated space switching on 32-bit inputs, where it achieves up to 9$\times$ speedup over scheme switching for $2^{14}$ rows. In contrast, direct comparison becomes impractical at 32-bit because it requires a large 32-bit prime $p$, and constructing the comparison polynomial incurs $O(p^2)$ cost, which is infeasible in practice. Furthermore, space switching completes in 1 hour 18 minutes for 56-bit inputs with $2^{12}$ rows, whereas HE$^3$DB, as a representative scheme-switching method, does not support bit-widths beyond 32-bit.

\subsection{Runtime Breakdown for Computation Stages}
\label{sec:expt breakdown}


We provide a detailed runtime analysis in Table~\ref{tab:breakdown}. Space switching links arithmetic in the number space $\mathbb{Z}_{p^r}$ with comparisons in the digit space $\mathbb{Z}_p$ via four stages: (1) \emph{reduction} from $\mathbb{Z}_{p^r}$ to $\mathbb{Z}_p$, (2) \emph{comparison} in $\mathbb{Z}_p$ via interpolation, (3) \emph{aggregation} of digit-wise results, and (4) \emph{modulus raising} from $\mathbb{Z}_p$ to $\mathbb{Z}_{p^r}$.

As shown in Table~\ref{tab:breakdown}, the relative contributions of these stages shift as the bit-width increases. For 8-bit inputs, the total runtime is 4.69s, with each stage contributing relatively evenly: reduction 1.89s (40\%), interpolation 1.65s (35\%), aggregation 0.44s (9\%), and modulus raising 0.71s (15\%). At 12 bits, the runtime grows to 22.99s, dominated by reduction at 12.06s (52\%) and interpolation at 6.45s (28\%), while aggregation and modulus raising remain modest at 2.21s (10\%) and 2.27s (10\%). With 16-bit inputs, the imbalance becomes stronger: reduction increases to 30.32s (42\% of the 72.56s total) and interpolation to 29.04s (40\%), while aggregation and modulus raising contribute only 2.51s (4\%) and 10.69s (15\%). At 20 bits, the runtime reaches 148.28s, where reduction alone takes 71.65s (48\%), interpolation 54.60s (37\%), aggregation 5.77s (4\%), and modulus raising 16.26s (11\%).  

Two main insights emerge from this breakdown. First, aggregation is consistently lightweight across all bit-widths, remaining under 10\% of the total computation. This reflects its dependence primarily on $r$ and a small number of multiplications, so its contribution stays limited even as the overall workload grows. Second, the dominant costs are the reduction and interpolation steps in the $\mathbb{Z}_{p}$ domain. Both grow significantly as the bit-width increases and together account for 75–85\% of the runtime, depending on the values of $p$ and $r$. Choosing a larger $p$ allows for a smaller $r$, meaning fewer digits to extract, but this also increases the cost of evaluating the comparison polynomial. In contrast, modulus raising shows moderate but steady growth, while aggregation remains negligible. Overall, the primary bottleneck of the proposed method lies in the reduction and interpolation steps, especially as the input size increases.

\subsection{Performance with Varying Bit-Width}
\label{sec:expt varying}

Table~\ref{tab:comp perf} reports the runtime of space switching for comparison operations across different bit-widths. Each row lists the chosen parameter set $(p,r)$, the ciphertext modulus size $\log Q$, the polynomial ring parameters $(m,n)$, the total runtime, and the amortized per-comparison runtime. Space switching achieves amortized latencies as low as 4.34 ms for 8-bit comparisons. For larger bit-widths, the amortized runtime remains within a few to tens of milliseconds, ranging from 7.66 ms for 12-bit inputs, 19.05 ms for 16-bit inputs, and 17.36 ms for 20-bit inputs. The amortized cost is strongly influenced by the order $d = \mathsf{Ord}(p)$ modulo $m$. With $n$ fixed, a smaller $d$ yields a larger SIMD packing factor $\ell$, which allows more comparisons to be processed in parallel, thereby reducing per-operation runtime. We also observe that as the bit-width of the input increases, both the total runtime and the amortized runtime generally increase. This is because larger inputs typically require a larger prime $p$ and exponent $r$. A larger $p$ leads to higher-degree polynomial $G_{p,r}(\cdot)$, as well as the interpolation of comparison polynomials over $\mathbb{Z}_p$. Meanwhile, $r$ determines the number of digits that must be processed: the larger $r$ is, the more the polynomial $G_{p,r}(\cdot)$ need to be evaluated, and the more computation is required in comparison operation over $\mathbb{Z}_{p}$.

While the overall runtime increases with the input bit-width, we can achieve better amortized runtime by choosing appropriate encryption parameters. For example, the amortized runtime on 20-bit input can be as little as 17.36 ms, comparable to or even faster than the 12-bit input. This implies that the choice of the $p$ and $r$ is non-trivial.  When the input bit-width is fixed, we can choose different sets of $(p,r)$. The larger the prime $p$ is, the smaller power $r$ is needed. Smaller $r$ means we need to evaluate the polynomial $G_{p,r}(\cdot)$ less frequently. On the other hand, the larger the prime $p$ is, the higher degree the interpolation polynomial and the polynomial $G_{p,r}$ have. This, in turn, increases the evaluation time. By choosing the proper $p$ and $r$, we can achieve a better runtime.

\subsection{Discussion}

\subsubsection*{CKKS-based approximate comparison}
We additionally measured the runtime of CKKS-based approximate comparison. In this setting, the CKKS inputs can be discretized and scaled to fit the approximation interval. We use the CKKS implementation in OpenFHE~\cite{OpenFHE}, with the LT operator implemented based on Cheon \textit{et al.}~\cite{cheon2020approx}. In CKKS, the number of SIMD slots is always half of the ring dimension $(n/2)$, which significantly reduces the amortized runtime per value. The amortized runtime for evaluating an LT operator is 0.58 ms for 8-bit inputs, 0.93 ms for 12-bit inputs, 1.39 ms for 16-bit inputs, and 1.99 ms for 20-bit inputs. CKKS is generally faster in this setting because it performs only a single polynomial evaluation for the comparison itself, and its output can be directly used in subsequent arithmetic operations without additional processing, whereas space switching requires additional steps, including reduction before digit comparison, followed by aggregation and modulus raising to return to the arithmetic domain. However, CKKS comparison remains approximate: its outputs are not exactly 0 or 1, but values close to them. Although discretization and scaling improve the practical behavior of CKKS-based comparison, they do not eliminate the main issue at the critical point of equality in a strict LT test, since the comparison function remains discontinuous there. In particular, when two inputs are equal, the approximation produces 0.5, whereas for a strict LT operation the expected result should be 0. This approximation error may propagate to subsequent arithmetic or aggregation operations, such as \textsf{SUM} in Figure~\ref{fig:db workflow}, and thus does not provide exact semantics after comparison. Moreover, achieving higher precision in CKKS requires higher-degree polynomials, which increases both runtime and multiplicative depth.

\subsubsection*{Functional bootstrapping approaches}
Functional-bootstrapping-based methods implement comparison through LUT evaluation during refresh and are most effective for relatively small input domains. Their cost scales with the LUT domain size, making them less suitable as bit-width grows. For example, Lee \textit{et al.}~\cite{lee2024functional} report 47.3 s for 9-bit inputs and 172.1 s for 12-bit inputs, and Liu \textit{et al.}~\cite{liu2023amortizedfunctionalbootstrapping} report 6.7 ms for 9-bit inputs and 39.1 ms for 12-bit inputs. These results suggest that functional bootstrapping can be competitive for relatively small input domains, but its cost grows more rapidly as bit-width increases. In contrast, our method increases from 4.34 ms at 8-bit to 7.66 ms at 12-bit. These results illustrate that our digit-level decomposition scales more favorably with bit-width than LUT-based functional bootstrapping. Therefore, while functional bootstrapping is a general approach, our work addresses a different and more scalable operating regime for exact comparison in leveled FV/BGV.
\section{Conclusion}
\label{s:conclusion}

In this work, we presented \emph{space switching}, a method to seamlessly integrate arithmetic and comparison operations within FV-style schemes such as BFV and BGV. While prior research shows that FV-style encryption can efficiently support arithmetic and accurately evaluate comparisons, performing these operations in succession has remained a challenge. Our approach resolves this by separating computations into their natural plaintext spaces, arithmetic in $\mathbb{Z}_{p^r}$ and comparisons in $\mathbb{Z}_p$, and introducing reduction and modulus-raising procedures to transition between them. We conducted extensive experiments across varying bit-widths, showing that space switching surpasses scheme switching by more than 20$\times$ for LT operator and by more than 100$\times$ compared to \emph{direct comparison}. For database workloads, space switching achieves up to 17$\times$ speedup over scheme switching and up to 15$\times$ improvement over direct comparison. We believe this work enhances the universality of FV-style schemes for general-purpose computation, making them more practical for applications such as privacy-preserving machine learning and genomic analysis. Extending space switching to support even more complex functions remains an important direction for future research.
\section{Acknowledgments}
We thank the anonymous reviewers for their valuable feedback and the shepherd for the invaluable guidance provided throughout the revision process. This material is based upon work supported by the National Science Foundation under Grant Nos. CCF-2523407 and CNS-2413232. Any opinions, findings, and conclusions or recommendations expressed in this material are those of the author(s) and do not necessarily reflect the views of the National Science Foundation.

\bibliographystyle{plain}
\bibliography{refs, lou}

\newpage

\appendices
\newpage 


\section{Meta-Review}

The following meta-review was prepared by the program committee for the 2026
IEEE Symposium on Security and Privacy (S\&P) as part of the review process as
detailed in the call for papers.

\subsection{Summary}
This paper proposes ``space switching'', a technique to efficiently support both arithmetic operations and comparisons within a single BFV/BGV homomorphic encryption scheme. The key idea is to convert between a number space $(\mathbb{Z}_{p^r})$ suited for arithmetic and a digit space $(\mathbb{Z}_p^r)$ suited for polynomial-interpolation-based comparison, using digit extraction techniques adapted from the bootstrapping literature and a modulus raising procedure. The approach is evaluated on primitive operators and database workloads, demonstrating significant speedups over scheme-switching-based methods.

\subsection{Scientific Contributions}
\begin{itemize}
\item Addresses a Long-Known Issue. 
\item Provides a Valuable Step Forward in an Established Field.
\end{itemize}

\subsection{Reasons for Acceptance}
\begin{enumerate}
\item Efficiently combining arithmetic and comparison operations within a single FHE scheme is a well-known problem. Prior approaches either require expensive scheme switching between different FHE schemes or rely on polynomial approximations that sacrifice exactness. This paper presents a principled approach that avoids both limitations, achieving up to 15$\times$ speedups over scheme-switching-based methods for comparison operations.
   
\item The paper repurposes digit decomposition techniques, previously used only in BFV/BGV bootstrapping, as a general-purpose tool for bridging arithmetic and comparison within a single scheme. While the underlying digit extraction algorithms build on prior bootstrapping literature, the conceptual insight that these techniques can be applied cost-effectively outside of bootstrapping is valuable and may encourage the community to reconsider the practical utility of digit decomposition more broadly.
\end{enumerate}

\end{document}